\documentclass{article}
\usepackage{amssymb}
\usepackage{amsmath}
\usepackage{amsfonts}
\usepackage{graphicx}

\newtheorem{theorem}{Theorem}

\newtheorem{corollary}[theorem]{Corollary}

\newtheorem{lemma}[theorem]{Lemma}

\newtheorem{remark}[theorem]{Remark}

\newenvironment{proof}[1][Proof]{\noindent\textbf{#1.} }{\ \rule{0.5em}{0.5em}}
\input{tcilatex}

\begin{document}

\title{Existence of orthogonal domain walls in B\'{e}nard-Rayleigh convection%
}
\author{ G\'{e}rard Iooss \\
{\footnotesize Laboratoire J.A.Dieudonn\'e, I.U.F., Universit\'e C\^ote
d'Azur, CNRS,}\\
{\footnotesize Parc Valrose, 06108 Nice cedex 2, France} \\
{\footnotesize iooss.gerard@orange.fr}}
\date{}
\maketitle

\begin{abstract}
In B\'{e}nard-Rayleigh convection we consider the pattern defect in
orthogonal domain walls connecting a set of convective rolls with another
set of rolls orthogonal to the first set. This is understood as an
heteroclinic orbit of a reversible system where the $x$ - coordinate plays
the role of time. This appears as a perturbation of the heteroclinic orbit
proved to exist in a reduced 6-dimensional system studied by a variational
method in \cite{BHI}, and studied analytically in \cite{Io23}. We then prove for a
given amplitude $\varepsilon^2$, and an imposed symmetry in coordinate $y$,
the existence of a one-parameter family of heteroclinic connections between
orthogonal sets of rolls, \ with wave numbers (different in general) which
are linked to an adapted "shift" of rolls parallel to the wall.
\end{abstract}


\section{Introduction}

\begin{remark}
This work slightly improves the results (see Theorem \ref{theor walls}) of a
previous version accepted for publication in JMFM (2024). The modifications
are consequences of an improvement of estimates obtained after revision  (see Theorem \ref{theorem}) in 
\cite{Io23}, which are extensively used here.
\end{remark}

The B\'{e}nard-Rayleigh convection problem is a classical problem in fluid
mechanics. It concerns the flow of a three-dimensional viscous fluid layer
situated between two horizontal parallel plates and heated from below. Upon
increasing the difference of temperature between the two plates, the simple
conduction state looses stability at a critical value of the temperature
difference corresponding to a critical value $\mathcal{R}_{c}$ of the
Rayleigh number. Beyond the instability threshold, a convective regime
develops in which patterns are formed, such as convective rolls, hexagons,
or squares \cite{kosch}. Observed patterns are often accompanied by defects
as for instance domain walls which occur between rolls with different
orientations. We refer to the works \cite{Bod,Man, Man-Pom}, and the
references therein, for experimental and analytical results, and detailed
descriptions of these patterns and defects.

Mathematically, the governing equations are the Navier-Stokes equations
coupled with an equation for the temperature, and completed by boundary
conditions at the two plates. Observed patterns are then found as particular
steady solutions of these equations. In \cite{HI Arma} and \cite{HI21b}
Haragus and Iooss handled the full governing Navier-Stokes-Boussinesq \
(N-S-B) equations and proved, for various boundary conditions, the existence
of symmetric domain walls in convection (however not yet observed
experimentally).

The existence of orthogonal domain walls (effectively observed
experimentally) has been studied formally by Manneville and Pomeau in \cite%
{Man-Pom}. In \cite{Boy-Vi} and \cite{Hu-Vi}, (this is named "planar 90$^{0}$
grain boundary separating two stripe domains of mutually perpendicular
orientations"), this is completed by the study of the dynamics of these
defects, function of the waves numbers of each set of rolls, however only on
a Swift-Hohenberg type of model ODE so that these previous works do not
start with the Navier-Stokes-Boussinesq system of equations, and just give
interesting asymptotic non rigorous results in the mathematical sense.

More recently Buffoni et al \cite{BHI} handle the full governing equations,
showing that the study leads to a small perturbation of the reduced system
of amplitude equations in $%
\mathbb{R}
^{6}$ , the same system as the one predicted in \cite{Man-Pom}: 
\begin{eqnarray}
A^{(4)} &=&A(1-A^{2}-gB^{2})  \label{reduced system} \\
B^{\prime \prime } &=&\varepsilon ^{2}B(-1+gA^{2}+B^{2}),  \notag
\end{eqnarray}%
where $\varepsilon ^{2}$ is the amplitude of rolls at infinities, and $g$ a
number, function of the Prandtl number of the flow. By a variational argument
Boris Buffoni et al \cite{BHI} prove the existence of an heteroclinic orbit,
for any $g>1,$ and $\varepsilon $ small enough, such that 
\begin{eqnarray*}
A_{\ast }(x)&>&0, \text{  }0<B_{\ast }(x) <1\\
(A_{\ast }(x),B_{\ast }(x)) &\rightarrow &\left\{ 
\begin{array}{c}
M_{-}=(1,0)\text{ as }x\rightarrow -\infty  \\ 
M_{+}=(0,1)\text{ as }x\rightarrow +\infty 
\end{array}%
\right. .
\end{eqnarray*}%
This orbit is expected to represent the connection between a set of
convecting rolls parallel to the $x$ direction, with a set of orthogonal
rolls. Unfortunately, this type of elegant proof does not allow to prove the
persistence of such heteroclinic curve under reversible perturbations of the
vector field, such that the one resulting from the full N-S-B system. Our
purpose here is to use the analytic results of \cite{Io23} for proving the
persistence of the above heteroclinic, hence applied to orthogonal domain
walls in B\'{e}nard-Rayleigh convection. It should be noticed that even
though the present analysis looks similar to the one made in \cite{HI Arma}
and \cite{HI21b}, it really needs serious adaptation since, here we loose
the symmetry of the wall defect, which plays an important role in \cite{HI
Arma} and \cite{HI21b}. Contrary to the symmetric case considered in \cite%
{HI Arma} and \cite{HI21b}, the size of the perturbation depends on $%
\varepsilon ,$ which appears also in the rescaled heteroclinic of system (%
\ref{reduced system}). This introduces lot of computations for controling
higher order terms (see section \ref{sect: estimates}). For obtaining steady
solutions of N-S-B system, we are led to consider the connection between
rolls of different wave numbers; we give the link between them and a
modulated "shift" of the system of rolls parallel to the wall, leading to a
one parameter set of solutions, for a fixed Rayleigh number slightly above
criticality, and a fixed Prandtl number. Contrary to the symmetric case, the
wave numbers of rolls at infinities need not be the same.

Section \ref{sect: reducedsyst} introduces the 8 dimensional system which
perturbs (\ref{reduced system}) and contains the full N-S-B system. Moreover
we give the final result in Theorem \ref{theor walls}. In section \ref%
{sect:setting} we introduce the new variables which tend exponentially
towards 0 at infinities, in such a way as to work in the weighted space $%
L_{\eta }^{2}.$ In section \ref{sect: estimates} we obtain estimates (in $%
L_{\eta }^{2})$ for solving in section \ref{sect: bifurc}, via a
Lyapunov-Schmidt reduction, the infinite-dimensional (in a function space)
part of the system. In subsection \ref{sect: finalbifurc} we solve the
one-dimensional remaining bifurcation equation leading to the result of
Theorem \ref{theor walls}. In Appendix \ref{App1} we indicate the normal
form found in \cite{BHI} and establish the perturbed system (\ref{reduced
syst a}). In Appendix \ref{App3} we give precisely the expression of the
equilibrium at $-\infty $ (rolls parallel to $x$ axis) and in Appendix \ref%
{App2} we give precisely the expression of the periodic solution at $+\infty 
$ (rolls parallel to the wall), giving a new analytic (necessary) proof for
the family of periodic solutions in the 1:1 resonance reversible bifurcation
problem (completing the former geometric proof of \cite{Io-Pe}).

\section{The reduced system\label{sect: reducedsyst}}

In \cite{BHI}, starting from a formulation of the steady governing N-S-B
equations as an infinite-dimensional dynamical system in which the
horizontal coordinate $x$ plays the role of evolutionary variable (spatial
dynamics), and looking for solutions periodic in $y$, a center manifold
reduction is performed, which leads to a $12$-dimensional reduced reversible
dynamical system, reducing to $8$-dimensional ($%
\mathbb{R}
^{4}\times 
\mathbb{C}
^{2}),$ after restricting to solutions with reflection symmetry $%
y\rightarrow -y$ (fixing the a priori free shift in the $y$ direction). A
normal form up to cubic order for this reduced system is obtained in \cite%
{BHI}. We may notice that $\varepsilon^2A_0$ and $\varepsilon^2B_0e^{ix/2%
\varepsilon}$ are respectively, after the scaling made in Appendix \ref{App1}%
, the principal parts of amplitudes (of order $\varepsilon^2$) of classical
convective rolls at $-\infty$ and $+\infty$.

After some calculations and rescaling (see (\ref{new reduced syst}) in
Appendix \ref{App1}) the perturbed system becomes%
\begin{eqnarray}
A_{0}^{(4)} &=&k_{-}A_{0}^{\prime \prime }+A_{0}(1-\frac{k_{-}^{2}}{4}%
-A_{0}^{2}-g|B_{0}|^{2})+\widehat{f},  \notag \\
B_{0}^{\prime \prime } &=&\varepsilon ^{2}B_{0}(-1+gA_{0}^{2}+|B_{0}|^{2})+%
\widehat{g}.  \label{reduced syst a}
\end{eqnarray}

Parameters are defined as (see Appendix \ref{App1})%
\begin{eqnarray*}
\varepsilon ^{4} &\sim &\mathcal{R}^{1/2}-\mathcal{R}_{c}^{1/2},\text{ }%
\mathcal{R}\text{ Rayleigh number,} \\
&&k_{c}(1+\varepsilon ^{2}k_{-})\text{ wave number in }y\text{ direction,}
\end{eqnarray*}

\begin{remark}
Notice that the system (\ref{reduced syst a}) becomes just system (\ref%
{reduced system}) for $k_{-}=\widehat{f}=\widehat{g}=0,$ and $B_{0}$ real.
\end{remark}

In (\ref{reduced syst a}) we have%
\begin{eqnarray*}
\widehat{f}(k_{-},\varepsilon ,\exp (\pm i\frac{x}{2\varepsilon }),X,Y,%
\overline{Y}) &=&\widehat{f_{0}}+\widehat{f_{1}} \\
\widehat{g}(k_{-},\varepsilon ,\exp (\pm i\frac{x}{2\varepsilon }),X,Y,%
\overline{Y}) &=&\widehat{g_{0}}+\widehat{g_{1}},
\end{eqnarray*}%
where 
\begin{eqnarray*}
X &=&(A_{0},A_{0}^{\prime },A_{0}^{\prime \prime },A_{0}^{\prime \prime
\prime })^{t}\in 
\mathbb{R}
^{4}, \\
Y &=&(B_{0},B_{0}^{\prime })^{t}\in 
\mathbb{C}
^{2}.
\end{eqnarray*}%
The dependency in $\exp (\pm i\frac{x}{2\varepsilon })$ of $\widehat{f}$ and 
$\widehat{g}$ comes from terms not in normal form, of degree at least 5 in $%
(X,Y)$ and the rescaling of the original amplitude $B$ of the rolls parallel
to the wall. In fact (see Appendix \ref{App1}) $B$ is rescaled as $%
\varepsilon ^{2}B_{0}e^{ix/2\varepsilon },$ where $x$ is the rescaled
coordinate. "Cubic" terms $\widehat{f_{0}},\widehat{g_{0}}$, are autonomous,
of the form%
\begin{eqnarray}
\widehat{f_{0}} &=&id_{1}\varepsilon A_{0}(B_{0}\overline{B_{0}}^{\prime }-%
\overline{B_{0}}B_{0}^{\prime })+\varepsilon ^{2}[\sigma
_{0}k_{-}A_{0}^{3}+d_{3}A_{0}^{\prime \prime }+d_{4}A_{0}^{2}A_{0}^{\prime
\prime }+d_{2}A_{0}A_{0}^{\prime 2}+d_{6}A_{0}|B_{0}^{\prime }|^{2}  \notag
\\
&&+d_{7}A_{0}^{\prime }(B_{0}\overline{B_{0}}^{\prime }+\overline{B_{0}}%
B_{0}^{\prime })+d_{5}A_{0}^{\prime \prime }|B_{0}|^{2}]+id_{8}\varepsilon
^{3}A_{0}^{\prime \prime }(B_{0}\overline{B_{0}}^{\prime }-\overline{B_{0}}%
B_{0}^{\prime })+\mathcal{O}(\varepsilon ^{4}),  \label{f0hat}
\end{eqnarray}%
\begin{eqnarray}
\widehat{g_{0}} &=&\varepsilon ^{3}[ic_{0}B_{0}^{\prime
}+ic_{1}B_{0}^{\prime }|A_{0}|^{2}+ic_{2}B_{0}^{\prime
}|B_{0}|^{2}+ic_{3}B_{0}^{2}\overline{B_{0}}^{\prime
}+ic_{9}B_{0}A_{0}A_{0}^{\prime }]  \label{g0hat} \\
&&+\varepsilon ^{4}[c_{4}B_{0}^{\prime }(B_{0}\overline{B_{0}}^{\prime }-%
\overline{B_{0}}B_{0}^{\prime })+c_{5}B_{0}A_{0}A_{0}^{\prime \prime
}+c_{6}B_{0}A_{0}^{\prime 2}+c_{7}B_{0}^{\prime }A_{0}A_{0}^{\prime }] 
\notag \\
&&+\varepsilon ^{5}[ic_{8}B_{0}A_{0}A_{0}^{\prime \prime \prime
}+ic_{7}B_{0}^{\prime }A_{0}A_{0}^{\prime \prime }+ic_{10}B_{0}^{\prime
}A_{0}^{\prime 2}+ic_{11}B_{0}A_{0}^{\prime }A_{0}^{\prime \prime }+\mathcal{%
O}(\varepsilon ^{6}),  \notag
\end{eqnarray}%
where oefficients $c_{j},$ $d_{j}$ are real (due to symmetries as seen in 
\cite{BHI} and Appendix \ref{App1}). Higher order terms, not in normal form
are non autonomous and such that%
\begin{eqnarray*}
\widehat{f_{1}} &=&\varepsilon ^{4}\mathcal{O[}|X|(|X|^{2}+|Y|^{2}+%
\varepsilon ^{4})^{2}], \\
\widehat{g_{1}} &=&\varepsilon ^{6}\mathcal{O[}%
(|X|^{2}+|Y|)(|X|^{2}+|Y|^{2}+\varepsilon ^{4})^{2}].
\end{eqnarray*}%
Moreover the system (\ref{reduced syst a}) commutes with the reversibility
symmetry $S_{1}:$%
\begin{equation*}
(x,A_{0},A_{0}^{\prime },A_{0}^{\prime \prime },A_{0}^{\prime \prime \prime
},B_{0},B_{0}^{\prime })\mapsto (-x,A_{0},-A_{0}^{\prime },A_{0}^{\prime
\prime },-A_{0}^{\prime \prime \prime },\overline{B_{0}},-\overline{%
B_{0}^{\prime }}),
\end{equation*}%
and we have the additional symmetry property (see \cite{BHI}) resulting from
the equivariance of the original system under the shift by half of a wave
length in the $y$ direction (fixing the symmetry $y\mapsto -y$):%
\begin{eqnarray*}
&&\text{r.h.s. of }A_{0}^{(4)}\text{ is odd in }X, \\
\text{r.h.s. of} &&B_{0}^{\prime \prime }\text{ is even in }X.
\end{eqnarray*}%
The estimates for non normal form terms $\widehat{f_{1}}$ and $\widehat{g_{1}%
},$ result from the property that they start at order 5, since the normal
form does not contain terms of degree 4 in $(X,Y),$ and from the inequality 
\begin{equation*}
(a+b)^{4}\leq 4(a^{2}+b^{2})^{2}\text{ for }a,b\in 
\mathbb{R}
.
\end{equation*}

\begin{remark}
Notice that the above reduction is valid for the three classical boundary
conditions for the B\'{e}nard-Rayleigh convection problem: rigid-rigid,
free-free, free-rigid. However in the case of rigid-rigid or free-free
boundary conditions, $Y=0$ is an invariant subspace (see \cite{BHI}), which
simplifies the estimate for $\widehat{g_{1}}.$
\end{remark}

\begin{remark}
Notice also that the high order terms $\widehat{f_{1}}$ and $\widehat{g_{1}}%
, $ of size $\mathcal{O}(\varepsilon ^{4})$ for $A_{0}^{(4)}$ and $\mathcal{O%
}(\varepsilon ^{6})$ for $B_{0}^{\prime \prime }$ are functions of $e^{\pm i%
\frac{x}{2\varepsilon }}.$ This is due to the fact that $\varepsilon
^{2}B_{0}e^{i\frac{x}{2\varepsilon }}$ is the original amplitude of the $Y$
mode (see (\ref{B0tilde}) in Appendix \ref{App1}).
\end{remark}

Let us give here the results obtained in \cite{Io23} for the system (\ref%
{reduced system}) and which are used in the calculations below:

\begin{theorem}
\label{theorem}Let us choose  $\frac{1}{3}\leq \delta \leq 1,$ and admit a
certain conjecture on a 4th order differential equation with boundary
conditions on a bounded interval, all being independent of $\varepsilon .$
Then for $\varepsilon $ small enough, the 3-dim unstable manifold of $M_{-}$
intersects transversally the 3-dim stable manifold of $M_{+},$ except for a
finite number of values of $\delta .$ The connecting curve $(A_{\ast
},B_{\ast })(x)$ which is obtained is the only curve for this intersection
going from $M_{-}$ towards $M_{+}$, and its dependency in parameters $%
(\varepsilon ,\delta )$ is analytic. In addition we have $B_{\ast }(x)$ and $%
B_{\ast }^{\prime }(x)>0$ on $(-\infty ,+\infty ).$ For $x\rightarrow
-\infty $ we have $(A_{\ast }-1,A_{\ast }^{\prime },A_{\ast }^{\prime \prime
},A_{\ast }^{\prime \prime \prime },B_{\ast },B_{\ast }^{\prime
})\rightarrow 0$ at least as $e^{\varepsilon \delta x},$ while for $%
x\rightarrow +\infty ,$ $(A_{\ast },A_{\ast }^{\prime },A_{\ast }^{\prime
\prime },A_{\ast }^{\prime \prime \prime })\rightarrow 0$ at least as $e^{-%
\sqrt{\frac{\delta }{2}}x},$ and $(B_{\ast }-1,B_{\ast }^{\prime
})\rightarrow 0$ at least as\ $e^{-\sqrt{2}\varepsilon x}.$
\end{theorem}

Moreover, choosing $0<\delta _{\ast }=\frac{1}{10}\delta ^{2/5}$ we have the
following useful estimates

\begin{corollary}
\label{corollary AA'} For $x\in (-\infty ,0]$ there exists $c>0$ independent
of $\varepsilon $ small enough, such that for the heteroclinic curve%
\begin{eqnarray*}
|A_{\ast }(x)-1| &\leq &ce^{2\varepsilon \delta _{\ast }x}, \\
|A_{\ast }^{\prime }(x)|+|A_{\ast }^{\prime \prime }(x)|+|A_{\ast }^{\prime
\prime \prime }(x)| &\leq &c\varepsilon ^{3/5}e^{2\varepsilon \delta _{\ast
}x}, \\
0 &<&B_{\ast }(x)\leq ce^{\varepsilon \delta _{\ast }x}, \\
0 &<&B_{\ast }^{\prime }(x)\leq c\varepsilon e^{\varepsilon \delta _{\ast
}x}.
\end{eqnarray*}
\end{corollary}

\begin{corollary}
\label{Corollary stab} For $x\in \lbrack 0,+\infty )$ there exists $c>0$
independent of $\varepsilon $ small enough, such that for the heteroclinic
curve%
\begin{eqnarray*}
|A_{\ast }^{(m)}(x)| &\leq &c\varepsilon ^{2/5}e^{-\varepsilon ^{1/5}\delta
_{\ast }x},\text{ }m=0,1,2,3, \\
|B_{\ast }(x)-1| &\leq &ce^{-\sqrt{2}\varepsilon x},\text{ }|B_{\ast
}^{\prime }(x)|\leq c\varepsilon e^{-\sqrt{2}\varepsilon x}.
\end{eqnarray*}
\end{corollary}

The above result is obtained in \cite{Io23} as follows: for system (\ref%
{reduced system}), from the equilibrium $M_{-}=(1,0)$ originates a
3-dimensional unstable invariant manifold and from the equilibrium $%
M_{+}=(0,1)$ originates a 3-dimensional stable invariant manifold. Both
manifolds lie on a 5 dimensional invariant manifold given by $\mathcal{W}%
_{g}=0$ where $\mathcal{W}_{g}$ is the first integral of (\ref{reduced
system}):%
\begin{equation}
\mathcal{W}_{g}=\varepsilon ^{2}(A_{0}^{\prime 2})^{\prime \prime
}-3\varepsilon ^{2}A_{0}^{\prime \prime 2}-|B_{0}^{\prime }|^{2}+\frac{%
\varepsilon ^{2}}{2}(A_{0}^{2}+|B_{0}|^{2}-1)^{2}+\varepsilon
^{2}(g-1)A_{0}^{2}|B_{0}|^{2}  \label{Wg}
\end{equation}
(this integral was known in \cite{Man-Pom}). The delicate point is then to
prove analytically that the two manifolds exist until they intersect
transversally, giving as a result the heteroclinic curve connecting $M_{-}$
to $M_{+}.$ The estimates in Corollaries above follow immediately from the
proof.

For the 8-dimensional perturbed system (\ref{reduced syst a}) we prove the
following :

\begin{theorem}
\label{theor walls} Except for a finite number of values of $g=1+\delta ^{2}$
and for $\varepsilon $ small enough, such that Theorem \ref{theorem}
applies, the heteroclinic solution connecting an equilibrium at $-\infty $
(representing convective rolls parallel to $x$ - axis and symmetric in
coordinate $y$) and a periodic solution at $+\infty $ (representing
convective rolls orthogonal to the previous ones, parallel to the wall),
exists as a family of orthogonal domain walls. Denoting by $\varepsilon ^{2}$
the amplitude of rolls at infinities, the wave number of rolls orthogonal to
the wall (resp. parallel to the wall) being $k_{c}(1+\varepsilon ^{2}k_{-})$
(resp. $k_{c}(1+\varepsilon ^{2}k_{+})),$ where $k_{c}$ is the critical wave
number, the result is the following: $k_{+}$ and $k_{-}$ are functions of $%
\varepsilon $ and of a parameter $\varphi ,$ such that%
\begin{eqnarray*}
|k_{+}(\varepsilon ,\varphi )| &\leq &c\varepsilon ^{2},\text{ } \\
k_{-}(\varepsilon ,\varphi ) &=&\mp \gamma _{1}\varepsilon ^{1+2/5}\exp
(-\varphi )+\mathcal{O}(\varepsilon ^{1+3/5}),\text{ with }\exp |\varphi
|\leq \varepsilon ^{-2/5}.
\end{eqnarray*}%
The parameter $\varphi $ is linked to the "shift" $z$ of rolls parallel to
the wall in such a way that%
\begin{equation*}
z=\gamma _{2}\varepsilon ^{1+2/5}(\exp \varphi \mp \exp (-\varphi ))+%
\mathcal{O}(\varepsilon ^{1+2/5}),
\end{equation*}%
where the numbers $\gamma _{1},\gamma _{2}$, the choice of $\pm $ in $z,$
and $k_{-}$ and the possibility to obtain $k_{-}=k_{+}$ only depend on $g$
and on the cubic coefficient $(d_{2}-d_{4})$ in the normal form found in 
\cite{BHI} (see Appendix \ref{App1} , (\ref{f0hat})), all being functions of
the Prandtl number.
\end{theorem}

\begin{remark}
Numbers $\gamma _{1}$ and $\gamma _{2}$ are given by%
\begin{equation*}
\gamma _{1}=2\sqrt{\frac{2|a_{5}|}{3}},\text{ }\gamma _{2}=\frac{\varepsilon^{1/5}}{2a_{1}}%
\sqrt{\frac{3|a_{5}|}{2}},
\end{equation*}%
where $a_{1}$ and $a_{5}$ are defined by (see Corollaries \ref{corollary AA'}%
, and \ref{Corollary stab} for the size of integrals)%
\begin{eqnarray*}
a_{1} &=&\int_{\mathbb{R}
}A_{\ast }^{\prime \prime 2}dx=\mathcal{O}(\varepsilon^{1/5}),\text{ \ }\mp =-sgn(a_{5}), \\
a_{5}\varepsilon ^{4/5} &=&(d_{2}-d_{4})\int_{%
\mathbb{R}
}A_{\ast }A_{\ast }^{\prime 3}dx.
\end{eqnarray*}%
Moreover, the equality $k_{+}=k_{-}=\mathcal{O}(\varepsilon ^{2})$ is
possible (for a suitable choice of $\varphi )$ if%
\begin{equation*}
d_{2}-d_{4}<0.
\end{equation*}
\end{remark}

\begin{remark}
The "shift" $z$ of rolls parallel to the wall is not a true shift $%
x\rightarrow x+z$, since $z$ influences non trivially the phase of $B$ (see
the function $w$ in (\ref{new var}) and Remark \ref{rmk w} giving the
principal part of $w$ proportionalto the cubic coefficient $c_{9}$ of the
normal form (\ref{g0hat})).
\end{remark}

\begin{remark}
The above family of solutions is invariant under the change $y\rightarrow -y$%
. The whole family may be shifted in $y$ direction, because of the
equivariance of the initial system under these shifts. The basic
heteroclinic solution for the truncated system (\ref{reduced system}) is
with a real amplitude $B,$ corresponding to a fixed position of rolls
parallel to the wall. After the rescaling, a "shift" $x\rightarrow x+z$
corresponds to a "shift" of order $z/\varepsilon $ of the original
coordinate. Choosing the parameter $\varphi $ such that $\exp |\varphi |=%
\mathcal{O}(\varepsilon ^{-1/5}),$which is allowed, we may obtain $z$ of
order $\varepsilon ,$ hence a significant "shift" (of order 1 in physical
space) for the rolls parallel to the wall.
\end{remark}

\begin{remark}
The wave numbers of the sets of rolls at $-\infty $ and at $+\infty $ differ
in general. This is a major difference with the symmetric case (of non
orthogonal walls) treated in \cite{HI Arma} and \cite{HI21b}.
\end{remark}

\begin{remark}
We might try to incorporate the 3 terms corresponding to coefficients $%
d_{2},d_{4}$ and $c_{9}$ of (\ref{f0hat}), (\ref{g0hat}) into a new first
integral as $\mathcal{W}_{g}$ (\ref{Wg}) (now with a complex $B_{0})$,
expecting to help in finding better estimates of the perturbed heteroclinic.
In fact, we cannot find such an integral, except if $d_{2}-d_{4}=c_{9}=0.$
This is indeed coherent with the necessity to look for different wave
numbers at infinities, as done in the present work.
\end{remark}

\begin{remark}
The coefficient $g=1+\delta ^{2}$ is function of the Prandtl number $%
\mathcal{P}$ and is the same as introduced and computed in (\cite{HI Arma}).
Values of $\delta $ such that $0.476\leq \delta $ include values obtained
for $\delta $ in the B\'{e}nard-Rayleigh convection problem. With
rigid-rigid, rigid-free, or free-free boundaries the minimum values of $g$
are respectively $(g_{\min }=1.227,$ $1.332,$ $1.423)$ corresponding to $%
\delta _{\min }=0.476,$ $0.576,$ $0.650.$ The restriction in Theorem \ref%
{theorem} corresponds to $1<g\leq 2.$ The eligible values for the Prandtl
number are respectively $\mathcal{P}>0.5308,>0.6222,>0.8078$. 
\begin{figure}[th]
\begin{center}
\includegraphics[width=4cm]{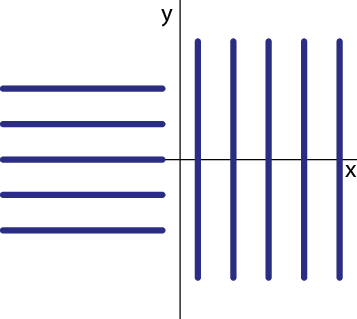}
\end{center}
\caption{Orthogonal domain wall}
\label{fig-wall}
\end{figure}
\end{remark}

\begin{remark}
Our method may be used for other physical problem displaying analogue
patterns, such as, for example at a fluid-ferro-fluid interface, as studied
in the symmetric case ("corner defect") by J.Horn in \cite{Horn}. More
generally, any physical problem leading to a normal form such as (\ref{full
syst}) (see Appendix \ref{App1}) introduces the 4 important coefficients $%
(g,d_{2},d_{4},c_{9})$ of the cubic normal form, and should, after
validation of the reduction, lead to a Theorem such as Theorem \ref{theor
walls}.
\end{remark}

\section{Setting of the perturbed system \label{sect:setting}}

\subsection{Solutions at infinities}

Since we leave now some freedom to the wave numbers, as well in the $y$
direction, as in the $x$ direction, the "end points" of the expected
heteroclinic are no longer $(1,0)$ at $-\infty ,$ and the circle $%
(0,e^{i\phi })$ at $+\infty .$ In fact the classical study of steady
convective rolls, shows that these should be respectively $(A_{0}^{(-\infty
)}(k_{-}),B_{0}^{(-\infty )}(k_{-}))$ and $(0,B_{0}^{(+\infty )}(\omega ,x))$
(see \cite{HIbook} section 4.3.3, or \cite{HI Arma} sections 2 and 6.2).
From Appendix \ref{App3} for the equilibrium at $-\infty $, we have 
\begin{eqnarray*}
(A_{0}^{(-\infty )})^{2} &=&1-\frac{k_{-}^{2}}{4}+\mathcal{\sigma }%
_{0}\varepsilon ^{2}k_{-}+\mathcal{O}(\varepsilon
^{2}|k_{-}|^{3}+\varepsilon ^{4}), \\
1-(A_{0}^{(-\infty )})\overset{def}{=}-\frac{\widetilde{\omega }_{-}^{2}}{2},%
\text{ with }\widetilde{\omega }_{-}^{2} &=&\frac{k_{-}^{2}}{4}-\mathcal{%
\sigma }_{0}\varepsilon ^{2}k_{-}+\mathcal{O}[k_{-}^{4}+\varepsilon
^{2}|k_{-}|^{3}+\varepsilon ^{4}], \\
B_{0}^{(-\infty )} &=&\mathcal{O}(\varepsilon ^{6}).
\end{eqnarray*}%
From Appendix \ref{App2} for the periodic solutions at $+\infty ,$ we have%
\begin{eqnarray*}
e^{i\frac{x}{2\varepsilon }}B_{0}^{(+\infty )}(\omega ,x) &=&r_{0}e^{i\omega
x}+\mathcal{O}(\varepsilon ^{6}),\text{ }A_{0}^{(+\infty )}=0, \\
\omega \overset{def}{=}\frac{1}{2\varepsilon }+\varepsilon \widetilde{\omega 
}_{+} &=&\frac{1+\varepsilon ^{2}k_{+}}{2\varepsilon }+\mathcal{O}%
(\varepsilon ^{7}),
\end{eqnarray*}%
\begin{equation*}
B_{0}^{(+\infty )}e^{-i\varepsilon \widetilde{\omega }_{+}x}=C_{0}^{(+\infty
)}+iD_{0}^{(+\infty )}
\end{equation*}%
\begin{eqnarray*}
r_{0}^{2} &=&1-\frac{k_{+}^{2}}{4}+\mathcal{O}(\varepsilon
^{2}|k_{+}|+\varepsilon ^{4})=1-\mathcal{O[}(\widetilde{|\omega }%
_{+}|+\varepsilon ^{2})^{2}], \\
C_{0}^{(+\infty )} &=&r_{0}+\mathcal{O}(\varepsilon ^{6}),\text{ oscil. part}%
(C_{0}^{(+\infty )})=\mathcal{O}(\varepsilon ^{6}), \\
D_{0}^{(+\infty )} &=&\mathcal{O}(\varepsilon ^{6}).
\end{eqnarray*}

\begin{remark}
The coefficient $\sigma _{0}$ introduced in the expression of $%
(A_{0}^{(-\infty )})^{2}$ depends on the Prandtl number.
\end{remark}

\begin{remark}
We may notice that in case the system has the symmetry $S_{0}$ representing $%
z\mapsto 1-z$ (OK for rigid-rigid, or free-free boundary conditions), then $%
B_{0}^{(-\infty )}=0,$ which simplifies computations (see Appendix \ref{App3}%
).
\end{remark}

\subsection{First change of variable}

Let us set%
\begin{equation*}
B_{0}e^{-i\varepsilon \widetilde{\omega }_{+}x}=C_{0}+iD_{0},
\end{equation*}%
then (\ref{reduced syst a}) becomes%
\begin{equation}
A_{0}^{(4)}=k_{-}A_{0}^{\prime \prime }+A_{0}[1-\frac{k_{-}^{2}}{4}%
-A_{0}^{2}-g(C_{0}^{2}+D_{0}^{2})]+f  \label{new basic1}
\end{equation}%
\begin{eqnarray}
C_{0}^{\prime \prime } &=&2\varepsilon \widetilde{\omega }_{+}D_{0}^{\prime
}+\varepsilon ^{2}C_{0}(-1+\widetilde{\omega }%
_{+}^{2}+gA_{0}^{2}+C_{0}^{2}+D_{0}^{2})+g_{r}  \label{new basic2} \\
D_{0}^{\prime \prime } &=&-2\varepsilon \widetilde{\omega }_{+}C_{0}^{\prime
}+\varepsilon ^{2}D_{0}(-1+\widetilde{\omega }%
_{+}^{2}+gA_{0}^{2}+C_{0}^{2}+D_{0}^{2})+g_{i}  \notag
\end{eqnarray}%
with 
\begin{equation*}
f=\widehat{f},\text{ \ }g_{r}+ig_{i}=\widehat{g}e^{-i\varepsilon \widetilde{%
\omega }_{+}x},
\end{equation*}%
and where the exponential factor disappears in the cubic part when we
replace $B_{0}$ by $(C_{0}+iD_{0})e^{i\varepsilon \widetilde{\omega }_{+}x}.$
Let us define

\begin{eqnarray*}
f &=&f_{0}(\varepsilon ,k_{-},X,Y,\overline{Y})+f_{1}(\omega x,\varepsilon
,k_{-},X,Y,\overline{Y}) \\
g_{r} &=&g_{r0}(\varepsilon ,X,Y,\overline{Y})+g_{r1}(\omega x,\varepsilon
,k_{-},X,Y,\overline{Y}) \\
g_{i} &=&g_{i0}(\varepsilon ,X,Y,\overline{Y})+g_{i1}(\omega x,\varepsilon
,k_{-},X,Y,\overline{Y}),
\end{eqnarray*}%
where $f_{0},g_{r0},g_{i0}$ come only from cubic terms of the normal form in
(\ref{reduced syst a}), and where $f_{1},g_{r1},g_{i1}$ are $2\pi -$periodic
in $\omega x$, smooth in their arguments, and satisfy estimates%
\begin{eqnarray*}
|f_{1}(\omega x,\varepsilon ,k_{-},X,Y,\overline{Y})| &\leq &c\varepsilon
^{4}|X|(|X|^{2}+|Y|^{2})^{2} \\
|g_{r1}(\omega x,\varepsilon ,k_{-},X,Y,\overline{Y})|+|g_{i1}(\omega
x,\varepsilon ,k_{-},X,Y,\overline{Y})| &\leq &c\varepsilon
^{6}(|X|^{2}+|Y|)(|X|^{2}+|Y|^{2})^{2},
\end{eqnarray*}%
with%
\begin{eqnarray*}
X &=&(A_{0},A_{0}^{\prime },A_{0}^{\prime \prime },A_{0}^{\prime \prime
\prime }) \\
Y &=&(C_{0}+iD_{0},C_{0}^{\prime }+iD_{0}^{\prime }).
\end{eqnarray*}%
Then we have from (\ref{f0hat}), (\ref{g0hat}):%
\begin{eqnarray}
f_{0} &=&d_{1}\varepsilon A_{0}(C_{0}D_{0}^{\prime }-D_{0}C_{0}^{\prime
})+\sigma _{0}\varepsilon ^{2}k_{-}A_{0}^{3}+d_{2}\varepsilon
^{2}A_{0}A_{0}^{\prime 2}+d_{3}\varepsilon ^{2}A_{0}^{\prime \prime }
\label{f0express1} \\
&&+d_{4}\varepsilon ^{2}A_{0}^{2}A_{0}^{\prime \prime }+d_{5}\varepsilon
^{2}A_{0}^{\prime \prime }(C_{0}^{2}+D_{0}^{2})+d_{6}\varepsilon
^{2}A_{0}(C_{0}^{\prime 2}+D_{0}^{\prime 2})+  \notag \\
&&+d_{7}\varepsilon ^{2}A_{0}^{\prime }(C_{0}C_{0}^{\prime
}+D_{0}D_{0}^{\prime })+d_{8}\varepsilon ^{3}A_{0}^{\prime \prime
}(C_{0}D_{0}^{\prime }-D_{0}C_{0}^{\prime })+\mathcal{O}(\varepsilon ^{4}), 
\notag
\end{eqnarray}%
\begin{eqnarray}
g_{r0}+ig_{i0} &=&i\varepsilon ^{3}(C_{0}^{\prime }+iD_{0}^{\prime
})[c_{0}+c_{1}A_{0}^{2}+c_{2}(C_{0}^{2}+D_{0}^{2})]  \label{gr0express} \\
&&+\varepsilon ^{3}c_{3}(C_{0}+iD_{0})(C_{0}D_{0}^{\prime
}-D_{0}C_{0}^{\prime })+i\varepsilon
^{3}c_{9}(C_{0}+iD_{0})A_{0}A_{0}^{\prime }  \notag \\
&&+\varepsilon ^{4}c_{4}(C_{0}^{\prime }+iD_{0}^{\prime
})(C_{0}D_{0}^{\prime }-D_{0}C_{0}^{\prime })+c_{5}\varepsilon
^{4}A_{0}A_{0}^{\prime \prime }(C_{0}+iD_{0})  \notag \\
&&+\varepsilon ^{4}[c_{6}A_{0}^{\prime
2}(C_{0}+iD_{0})+c_{7}A_{0}A_{0}^{\prime }(C_{0}^{\prime }+iD_{0}^{\prime })]
\notag \\
&&+i\varepsilon ^{5}(C_{0}^{\prime }+iD_{0}^{\prime
})(c_{7}A_{0}A_{0}^{\prime \prime }+c_{10}A_{0}^{\prime 2})  \notag \\
&&+i\varepsilon ^{5}(C_{0}+iD_{0})(c_{8}A_{0}A_{0}^{\prime \prime \prime
}+c_{11}A_{0}^{\prime }A_{0}^{\prime \prime })+\mathcal{O}(\varepsilon ^{6}).
\notag
\end{eqnarray}

Now, let us set a first change of variables

\begin{eqnarray}
A_{0} &=&A_{\ast }+\widetilde{A_{0}}  \notag \\
C_{0} &=&B_{\ast }+\widetilde{C_{0}}  \label{firstvarchange} \\
D_{0} &=&\widetilde{D_{0}}  \notag
\end{eqnarray}%
where we observe that we expect 
\begin{eqnarray*}
\widetilde{A_{0}}\underset{x=-\infty }{\rightarrow }A_{0}^{(-\infty )}-1 &=&-%
\frac{\widetilde{\omega }_{-}^{2}}{2},\text{ } \\
C_{0}+iD_{0}\underset{x=-\infty }{\rightarrow }C_{0}^{(-\infty )}
&=&B_{0}^{(-\infty )}=\mathcal{O}(\varepsilon ^{6}), \\
\widetilde{C_{0}}+i\widetilde{D_{0}}\underset{x=+\infty }{\rightarrow }%
C_{0}^{(+\infty )}+iD_{0}^{(+\infty )}-1 &\sim &-\frac{(\widetilde{\omega }%
_{+}+\mathcal{O}(\varepsilon ^{2}))^{2}}{2}.
\end{eqnarray*}%
Then (\ref{new basic1},\ref{new basic2}) becomes the "perturbed system"%
\begin{equation}
\mathcal{M}_{g}(\widetilde{A_{0}},\widetilde{C_{0}})=\binom{-k_{-}(A_{\ast
}^{\prime \prime }+\widetilde{A_{0}}^{\prime \prime })+\frac{k_{-}^{2}}{4}%
(A_{\ast }+\widetilde{A_{0}})+\widetilde{\phi _{0}}}{\frac{2\widetilde{%
\omega }_{+}}{\varepsilon }\widetilde{D_{0}}^{\prime }+\widetilde{\omega }%
_{+}^{2}(B_{\ast }+\widetilde{C_{0}})+\widetilde{\psi _{0r}}},
\label{perturbSysta}
\end{equation}%
\begin{equation}
\mathcal{L}_{g}\widetilde{D_{0}}=-\frac{2\widetilde{\omega }_{+}}{%
\varepsilon }(B_{\ast }^{\prime }+\widetilde{C_{0}}^{\prime })+\widetilde{%
\omega }_{+}^{2}\widetilde{D_{0}}+\widetilde{\psi _{0i}},
\label{perturbSystb}
\end{equation}%
where linear operators $\mathcal{M}_{g}$ and $\mathcal{L}_{g}$ are defined as%
\begin{equation}
\mathcal{M}_{g}\left( 
\begin{array}{c}
A \\ 
C%
\end{array}%
\right) =\left( 
\begin{array}{c}
-A^{(4)}+(1-3A_{\ast }^{2}-gB_{\ast }^{2})A-2gA_{\ast }B_{\ast }C \\ 
\frac{1}{\varepsilon ^{2}}C^{\prime \prime }+(1-gA_{\ast }^{2}-3B_{\ast
}^{2})C-2gA_{\ast }B_{\ast }A%
\end{array}%
\right) ,  \label{def Mg}
\end{equation}%
\begin{equation}
\mathcal{L}_{g}D=\frac{1}{\varepsilon ^{2}}D^{\prime \prime }+(1-gA_{\ast
}^{2}-B_{\ast }^{2})D,  \label{def Lg}
\end{equation}%
and where $\widetilde{\phi _{0}},\widetilde{\psi _{0r}},\widetilde{\psi _{0i}%
}$ are smooth functions of $(\omega x,\varepsilon ,k_{-},\widetilde{\omega }%
_{+},\widetilde{X},\widetilde{Y})$ where%
\begin{eqnarray*}
\widetilde{X} &=&(\widetilde{A_{0}},\widetilde{A_{0}}^{\prime },\widetilde{%
A_{0}}^{\prime \prime },\widetilde{A_{0}}^{\prime \prime \prime }) \\
\widetilde{Y} &=&(\widetilde{C_{0}},\widetilde{D_{0}},\widetilde{C_{0}}%
^{\prime },\widetilde{D_{0}}^{\prime })
\end{eqnarray*}%
\begin{eqnarray*}
\widetilde{\phi _{0}} &=&\widetilde{\phi _{00}}(\varepsilon ,k_{-},%
\widetilde{X},\widetilde{Y})+\widetilde{\phi _{01}}(\omega x,\varepsilon
,k_{-},\widetilde{X},\widetilde{Y}) \\
\widetilde{\psi _{0r}} &=&\widetilde{\psi _{0r0}}(\varepsilon ,k_{-},%
\widetilde{X},\widetilde{Y})+\widetilde{\psi _{0r1}}(\omega x,\varepsilon
,k_{-},\widetilde{X},\widetilde{Y}) \\
\widetilde{\psi _{0i}} &=&\widetilde{\psi _{0i0}}(\varepsilon ,k_{-},%
\widetilde{X},\widetilde{Y})+\widetilde{\psi _{0i1}}(\omega x,\varepsilon
,k_{-},\widetilde{X},\widetilde{Y})
\end{eqnarray*}%
\begin{eqnarray}
|\widetilde{\phi _{01}}(\omega x,\varepsilon ,k_{-},\widetilde{X},\widetilde{%
Y})| &\leq &c\varepsilon ^{4}  \label{estimrest nonnormal} \\
|\widetilde{\psi _{0r1}}(\omega x,\varepsilon ,k_{-},\widetilde{X},%
\widetilde{Y})|+|\widetilde{\psi _{0i1}}(\omega x,\varepsilon ,k_{-},%
\widetilde{X},\widetilde{Y})| &\leq &c\varepsilon ^{4}.  \notag
\end{eqnarray}%
More precisely, we have%
\begin{eqnarray}
\widetilde{\phi _{00}}(\varepsilon ,k_{-},\widetilde{X},\widetilde{Y})
&=&3A_{\ast }\widetilde{A_{0}}^{2}+\widetilde{A_{0}}^{3}+2gB_{\ast }%
\widetilde{A_{0}}\widetilde{C_{0}}  \label{phi00tildea} \\
&&+g(A_{\ast }+\widetilde{A_{0}})(\widetilde{C_{0}}^{2}+\widetilde{D_{0}}%
^{2})+f_{00},  \notag
\end{eqnarray}%
\begin{eqnarray}
\widetilde{\psi _{0r0}}(\varepsilon ,k_{-},\widetilde{X},\widetilde{Y})
&=&2gA_{\ast }\widetilde{A_{0}}\widetilde{C_{0}}+gB_{\ast }\widetilde{A_{0}}%
^{2}+2B_{\ast }\widetilde{C_{0}}^{2}+g\widetilde{A_{0}}^{2}\widetilde{C_{0}}
\label{psi0rtildea} \\
&&+(B_{\ast }+\widetilde{C_{0}})(\widetilde{C_{0}}^{2}+\widetilde{D_{0}}%
^{2})+g_{00r},  \notag
\end{eqnarray}%
\begin{eqnarray}
\widetilde{\psi _{0i0}}(\varepsilon ,k_{-},\widetilde{X},\widetilde{Y})
&=&2gA_{\ast }\widetilde{A_{0}}\widetilde{D_{0}}+2B_{\ast }\widetilde{C_{0}}%
\widetilde{D_{0}}+g\widetilde{A_{0}}^{2}\widetilde{D_{0}}  \notag \\
&&+\widetilde{D_{0}}(\widetilde{C_{0}}^{2}+\widetilde{D_{0}}^{2})+g_{00i},
\label{psi0itilde}
\end{eqnarray}%
and in using Theorem \ref{theorem}, Corollaries \ref{corollary AA'} and \ref%
{Corollary stab}, and assuming%
\begin{equation}
|\widetilde{X}|\leq 1,|\widetilde{Y}|\leq 1,|\widetilde{D_{0}}|\leq
\varepsilon ,  \label{basic estim1}
\end{equation}%
\begin{eqnarray*}
f_{00} &=&\sigma _{0}\varepsilon ^{2}k_{-}A_{\ast }^{3}+\mathcal{O}%
[\varepsilon ^{2+3/5}e^{\varepsilon \delta _{\ast }x}\chi _{(-\infty
,0)}+\varepsilon ^{2+2/5}e^{-\varepsilon ^{1/5}\delta _{\ast }x}\chi
_{(0,\infty )}+\varepsilon ^{2}(|\widetilde{X}|+|\widetilde{Y}|)+\varepsilon
|\widetilde{D_{0}}^{\prime }|], \\
g_{00r} &=&\mathcal{O}[\varepsilon ^{2+3/5}e^{\varepsilon \delta _{\ast
}x}\chi _{(-\infty ,0)}+\varepsilon ^{2+4/5}e^{-\varepsilon ^{1/5}\delta
_{\ast }x}\chi _{(0,\infty )}+\varepsilon ^{2}(|\widetilde{X}|+|\widetilde{Y}%
|)+\varepsilon |\widetilde{D_{0}}^{\prime }|], \\
g_{00i} &=&\mathcal{O}[\varepsilon ^{1+3/5}e^{\varepsilon \delta _{\ast
}x}\chi _{(-\infty ,0)}+\varepsilon ^{1+4/5}e^{-\varepsilon ^{1/5}\delta
_{\ast }x}\chi _{(0,\infty )}+\varepsilon (|\widetilde{X}|+|\widetilde{Y}%
|)+\varepsilon (|\widetilde{C_{0}}^{\prime }|],
\end{eqnarray*}%
where $f_{00}$ and $g_{00r}+ig_{00i}$ are smooth functions which come from
the rest of the cubic normal form written in (\ref{f0express1},\ref%
{gr0express})) and $\chi _{(-\infty ,0)}$ and $\chi _{(0,\infty )}$ are the
characteristic functions on the corresponding intervals.

\begin{remark}
\label{remark estim f00A'*}We notice that the estimates for the main terms
independent of $\widetilde{X},\widetilde{Y}$ come from%
\begin{eqnarray*}
\text{for }f_{00} &:&\text{ }\sigma _{0}\varepsilon ^{2}k_{-}A_{\ast
}^{3}+d_{2}\varepsilon ^{2}A_{\ast }A_{\ast }^{\prime 2}+d_{3}\varepsilon
^{2}A_{\ast }^{\prime \prime }+d_{4}\varepsilon ^{2}A_{\ast }^{2}A_{\ast
}^{\prime \prime }+d_{5}\varepsilon ^{2}A_{\ast }^{\prime \prime }B_{\ast
}^{2}, \\
\text{for }g_{00r} &:&\text{ }c_{5}\varepsilon ^{2}A_{\ast }A_{\ast
}^{\prime \prime }B_{\ast }+c_{6}\varepsilon ^{2}A_{\ast }^{\prime 2}B_{\ast
}+c_{7}\varepsilon ^{2}A_{\ast }A_{\ast }^{\prime }B_{\ast }^{\prime } \\
\text{for }g_{00i} &:&\text{ }\varepsilon B_{\ast }^{\prime
}(c_{0}+c_{1}A_{\ast }^{2}+c_{2}B_{\ast }^{2})+\varepsilon c_{9}B_{\ast
}A_{\ast }A_{\ast }^{\prime }.
\end{eqnarray*}%
Moreover, notice that, below, we need to compute $\int f_{00}A_{\ast
}^{\prime }dx,$ $\int g_{00r}B_{\ast }^{\prime }dx,$ $\int g_{00i}B_{\ast
}dx,$ which, for terms independent of $\widetilde{X},\widetilde{Y}$ leads to%
\begin{eqnarray*}
\text{for }\int f_{00}A_{\ast }^{\prime }dx &=&-\frac{\sigma _{0}\varepsilon
^{2}k_{-}}{4}+\varepsilon ^{2}\int (d_{2}A_{\ast }A_{\ast }^{\prime
3}+d_{4}A_{\ast }^{2}A_{\ast }^{\prime }A_{\ast }^{\prime \prime })dx+%
\mathcal{O}(\varepsilon ^{3}) \\
&=&-\frac{\sigma _{0}\varepsilon ^{2}k_{-}}{4}+\mathcal{O}(\varepsilon
^{2+4/5}), \\
\text{for }\int g_{00r}B_{\ast }^{\prime }dx &\sim &\varepsilon ^{2}\int_{%
\mathbb{R}
}c_{5}A_{\ast }A_{\ast }^{\prime \prime }B_{\ast }B_{\ast }^{\prime
}dx+\varepsilon ^{2}\int_{%
\mathbb{R}
}c_{6}A_{\ast }^{\prime 2}B_{\ast }B_{\ast }^{\prime }dx=\mathcal{O}%
(\varepsilon ^{3+1/5}), \\
\text{for }\int g_{00i}B_{\ast }dx &=&\varepsilon (\frac{c_{0}}{2}+\frac{%
c_{2}}{4})+\varepsilon (c_{1}-c_{9})\int_{%
\mathbb{R}
}A_{\ast }^{2}B_{\ast }B_{\ast }^{\prime }=\mathcal{O}(\varepsilon ),
\end{eqnarray*}%
where we notice 
\begin{eqnarray*}
\int A_{\ast }^{\prime \prime }A_{\ast }^{\prime }dx &=&0,\text{ }%
\varepsilon ^{2}\int A_{\ast }^{\prime \prime }B_{\ast }^{2}A_{\ast
}^{\prime }dx=-\varepsilon ^{2}\int A_{\ast }^{\prime 2}B_{\ast }B_{\ast
}^{\prime }dx=\mathcal{O}(\varepsilon ^{3}), \\
\int (d_{2}A_{\ast }A_{\ast }^{\prime 3}+d_{4}A_{\ast }^{2}A_{\ast }^{\prime
}A_{\ast }^{\prime \prime })dx &=&(d_{2}-d_{4})\int A_{\ast }A_{\ast
}^{\prime 3}dx=\mathcal{O}(\varepsilon ^{1/2}), \\
\int_{%
\mathbb{R}
}A_{\ast }A_{\ast }^{\prime \prime }B_{\ast }B_{\ast }^{\prime }dx &=&-\int_{%
\mathbb{R}
}[A_{\ast }^{\prime 2}B_{\ast }B_{\ast }^{\prime }+A_{\ast }A_{\ast
}^{\prime }(B_{\ast }B_{\ast }^{\prime })^{\prime }]dx,
\end{eqnarray*}%
taking care of the convergence in $e^{\varepsilon \delta _{\ast }x}$ (resp $%
e^{-\varepsilon ^{1/5}\delta _{\ast }x})$ at $-\infty $ (resp at $+\infty ),$
which implies a division by $\varepsilon $ in the integral on $(-\infty ,0)$
(resp. by $\varepsilon ^{1/5}$ in the integral on $(0,+\infty )).$
\end{remark}

\subsection{Second change of variables}

Before solving the system we need to change variables so that the variables
and the right hand side of (\ref{perturbSysta},\ref{perturbSystb}) tend
towards $0$ at infinity. Let us denote 
\begin{eqnarray*}
\widetilde{X}^{(-\infty )} &=&(A_{0}^{(-\infty )}-1,0,0,0)=(\mathcal{O}(%
\widetilde{\omega }_{-}^{2}),0,0,0) \\
\widetilde{Y}^{(-\infty )} &=&(C_{0}^{(-\infty )},0,0,0)=(\mathcal{O}%
(\varepsilon ^{6}),0,0,0), \\
\widetilde{X}^{(+\infty )} &=&0 \\
\widetilde{Y}^{(+\infty )} &=&(C_{0}^{(+\infty )}-1,D_{0}^{(+\infty
)},C_{0}^{(+\infty )\prime },D_{0}^{(+\infty )\prime })=[\mathcal{O}((%
\widetilde{\omega }_{+}+\varepsilon ^{2})^{2}),\mathcal{O}(\varepsilon ^{6}),%
\mathcal{O}(\varepsilon ^{5}),\mathcal{O}(\varepsilon ^{5})],
\end{eqnarray*}%
then, taking care in (\ref{new basic1},\ref{new basic2}), of the forms of $f$%
, $g_{r},$ $g_{i},$ we notice that the limit terms in the right hand side of
(\ref{perturbSysta},\ref{perturbSystb}) as $x\rightarrow -\infty $ are%
\begin{eqnarray*}
&&\frac{k_{-}^{2}}{4}A_{0}^{(-\infty )}+\widetilde{\phi _{0}}(\omega
x,\varepsilon ,k_{-},\widetilde{X}^{(-\infty )},\widetilde{Y}^{(-\infty )})%
\text{ exp limit as }e^{\varepsilon \delta _{\ast }x}\text{ (see }f_{00}), \\
&&\widetilde{\omega }_{+}^{2}C_{0}^{(-\infty )}+\widetilde{\psi _{0r}}%
(\omega x,\varepsilon ,k_{-},\widetilde{X}^{(-\infty )},\widetilde{Y}%
^{(-\infty )})\text{ exp limit as }e^{\varepsilon \delta _{\ast }x}\text{
(see }g_{00r}) \\
&&\widetilde{\psi _{0i}}(\omega x,\varepsilon ,k_{-},\widetilde{X}^{(-\infty
)},\widetilde{Y}^{(-\infty )})\text{ exp limit as }e^{\varepsilon \delta
_{\ast }x}\text{ (as }B_{\ast }^{\prime }\text{ and see }g_{00i}).
\end{eqnarray*}%
The limit terms of the right hand side of (\ref{perturbSysta},\ref%
{perturbSystb}) as $x\rightarrow +\infty $ is%
\begin{eqnarray*}
&&0\text{ exp limit as }e^{-\varepsilon ^{1/5}\delta _{\ast }x}\text{ (as }%
A_{\ast }) \\
&&\frac{2\widetilde{\omega }_{+}}{\varepsilon }(D_{0}^{(+\infty )})^{\prime
}+\widetilde{\omega }_{+}^{2}C_{0}^{(+\infty )}+\widetilde{\psi _{0r}}%
(\omega x,\varepsilon ,k_{-},0,\widetilde{Y}^{(+\infty )})\text{ exp limit
as }e^{-\varepsilon ^{1/5}\delta _{\ast }x}\text{ (see }g_{00r}), \\
&&-\frac{2\widetilde{\omega }_{+}}{\varepsilon }(C_{0}^{(+\infty )})^{\prime
}+\widetilde{\omega }_{+}^{2}D_{0}^{(+\infty )}+\widetilde{\psi _{0i}}%
(\omega x,\varepsilon ,k_{-},0,\widetilde{Y}^{(+\infty )})\text{ exp limit
as }e^{-\varepsilon \sqrt{2}x}\text{ (see }g_{00i}).
\end{eqnarray*}

Let us make a second change of variables as%
\begin{eqnarray}
\widetilde{A_{0}} &=&\alpha _{-}\chi _{-}+\widehat{A_{0}}  \notag \\
\widetilde{C_{0}} &=&\beta _{-}\chi _{-}+\beta _{+}\chi _{+}+\widehat{C_{0}},
\label{changevariable} \\
\widetilde{D_{0}} &=&\gamma _{+}\chi _{+}+\widehat{D_{0}},  \notag
\end{eqnarray}%
with (in using Appendix \ref{App3} and (\ref{r0(k+,eps)}) in Appendix \ref%
{App2})%
\begin{eqnarray}
\alpha _{-} &=&(A_{0}^{(-\infty )}-1)=\mathcal{-}\widetilde{\omega }%
_{-}^{2}/2,\text{ }\beta _{-}=B_{0}^{(-\infty )},  \label{defalpha-} \\
\beta _{+} &=&(C_{0}^{(+\infty )}(\omega x)-1),\text{ }\gamma
_{+}=D_{0}^{(+\infty )}(\omega x),  \notag
\end{eqnarray}%
\begin{equation}
\text{const part of }\beta _{+}\overset{def}{=}\beta _{+}^{(c)}=-\frac{%
\widetilde{\omega }_{+}^{2}}{2}+\frac{\sigma _{1}\varepsilon ^{2}\widetilde{%
\omega }_{+}}{2}+\frac{\sigma _{2}\varepsilon ^{4}}{2}+\mathcal{O[}(|%
\widetilde{\omega }_{+}|+\varepsilon ^{2})^{4}],  \label{defbeta+c}
\end{equation}%
and where $\chi _{-}$ and $\chi _{+}$ are smooth functions, such that%
\begin{eqnarray*}
\chi _{-} &=&1\text{ for }x\in (-\infty ,-1), \\
&=&0\text{ for }x>0 \\
0 &<&\chi _{-}<1\text{ for }x\in (-1,0),
\end{eqnarray*}%
\begin{eqnarray*}
\chi _{+} &=&1\text{ for }x\in (1,\infty ), \\
&=&0\text{ for }x<0 \\
0 &<&\chi _{+}<1\text{ for }x\in (0,1),
\end{eqnarray*}%
such that 
\begin{equation*}
(\widehat{A_{0}},\widehat{C_{0}},\widehat{D_{0}})\rightarrow 0\text{ as }%
|x|\rightarrow \infty .
\end{equation*}

\subsection{Properties of linear operators $\mathcal{M}_{g}$ and $\mathcal{L}%
_{g}$ (defined in (\protect\ref{def Mg},\protect\ref{def Lg}))}

We now give a precise definition of the function spaces where we will solve
the problem with respect to $(\widehat{A_{0}},\widehat{C_{0}},\widehat{D_{0}}%
).$ Indeed, let us define the Hilbert spaces%
\begin{equation*}
L_{\eta }^{2}=\{u;u(x)e^{\eta |x|}\in L^{2}(%
\mathbb{R}
)\},
\end{equation*}%
\begin{eqnarray*}
\mathcal{D}_{0} &=&\{(A,C)\in H_{\eta }^{4}\times H_{\eta }^{2};A\in H_{\eta
}^{4},C\in \mathcal{D}_{1}\} \\
\mathcal{D}_{1} &=&\{C\in H_{\eta }^{2};\varepsilon ^{-2}||C^{\prime \prime
}||_{L_{\eta }^{2}}+\varepsilon ^{-1}||C^{\prime }||_{L_{\eta
}^{2}}+||C||_{L_{\eta }^{2}}\overset{def}{=}||C||_{\mathcal{D}_{1}}<\infty \}
\end{eqnarray*}%
equiped with natural scalar products. Then we have the following result
(proved in \cite{Io23}):

\begin{lemma}
\label{Lem Linearoperator}Except maybe for a set of isolated values of $g,$
the kernel of $\mathcal{M}_{g}$ in $L_{\eta }^{2}$ is one dimensional,
spanned by $(A_{\ast }^{\prime },B_{\ast }^{\prime }),$ and its range has
codimension 1, $L^{2}$- orthogonal to $(A_{\ast }^{\prime },B_{\ast
}^{\prime }).$ $\mathcal{M}_{g}$ has a pseudo-inverse acting from $L_{\eta
}^{2}$ to $\mathcal{D}_{0}$ for any $\eta >0$ small enough, with bound
independent of $\varepsilon .$

The operator $\mathcal{L}_{g}$ has a trivial kernel, and its range which has
codimension 1, is $L^{2}$- orthogonal to $B_{\ast }$ ($B_{\ast }\notin
L^{2}).$ $\mathcal{L}_{g}$ has a pseudo-inverse acting respectively from $%
L_{\eta }^{2}$ to $\mathcal{D}_{1}$ for $\eta >0$ small enough, with bound
independent of $\varepsilon .$
\end{lemma}

\begin{remark}
We might expect a two-dimensional kernel since we have a "circle" of
heteroclinics. The one-dimensional kernel of $\mathcal{M}_{g}$ is the usual
one, while we also have $\mathcal{L}_{g}B_{\ast }=0.$ However $B_{\ast
}\notin L_{\eta }^{2}$ so that the kernel of $\mathcal{L}_{g}$ is $\{0\},$
and we pay this by a codimension one range for $\mathcal{L}_{g}.$ This is
explicitely computed in \cite{Io23}.
\end{remark}

\section{Estimates for the right hand sides of $\mathcal{M}_{g}(\widehat{%
A_{0}},\widehat{C_{0}})$ and $\mathcal{L}_{g}\widehat{D_{0}}\label{sect:
estimates}$}

After the second change of variables (\ref{changevariable}) the remaining
terms in the right hand side of $\mathcal{M}_{g}(\widehat{A_{0}},\widehat{%
C_{0}})$ and $\mathcal{L}_{g}\widehat{D_{0}}$ coming from 
\begin{equation*}
\widetilde{\phi _{01}}(\omega x,\varepsilon ,k_{-},\widetilde{X},\widetilde{Y%
}),\text{ }\widetilde{\psi _{0r1}}(\omega x,\varepsilon ,k_{-},\widetilde{X},%
\widetilde{Y}),\text{ }\widetilde{\psi _{0i1}}(\omega x,\varepsilon ,k_{-},%
\widetilde{X},\widetilde{Y})
\end{equation*}%
now cancel for $(\widehat{X},\widehat{Y},\widehat{\overline{Y}})=0,$ they
are then estimated in $L_{\eta }^{2}$ by 
\begin{equation}
\mathcal{O}(\varepsilon ^{4}(||(\widehat{A_{0}},\widehat{C_{0}})||_{\mathcal{%
D}_{0}}+||\widehat{D_{0}}||_{\mathcal{D}_{1}}),  \label{estim rest 1}
\end{equation}%
provided that the following condition%
\begin{equation}
|\widehat{A_{0}}(x)|+|\widehat{A_{0}^{\prime }}(x)|+|\widehat{A_{0}^{\prime
\prime }}(x)|+|\widehat{A_{0}^{\prime \prime \prime }}(x)|+|\widehat{C_{0}}%
(x)|+|\widehat{C_{0}^{\prime }}(x)|+|\widehat{D_{0}}(x)|+|\widehat{%
D_{0}^{\prime }}(x)|<<1  \label{condition1}
\end{equation}%
holds. We need to check this condition at the end of subsection \ref{sect:
finalbifurc}. The unknowns in the problem are now%
\begin{equation*}
(\widehat{A_{0}},\widehat{C_{0}})\in \mathcal{D}_{0},\text{ }\widehat{D_{0}}%
\in \mathcal{D}_{1},\text{ }(k_{-},\widetilde{\omega }_{+})\in 
\mathbb{R}
^{2},
\end{equation*}%
and $\varepsilon $ is supposed to be small enough. In the following we use
extensively the estimates (see (\ref{defalpha-},\ref{defbeta+c}))%
\begin{eqnarray*}
\alpha _{-} &=&\mathcal{O}(|k_{-}|+\varepsilon ^{2})^{2},\text{ \ }\beta
_{+}=\mathcal{O}(|\widetilde{\omega }_{+}|+\varepsilon ^{2})^{2}, \\
\beta _{-} &=&\mathcal{O}(\varepsilon ^{6}),\text{ oscil part }(\beta _{+})=%
\mathcal{O}(\varepsilon ^{6}),\text{ }\gamma _{+}=\mathcal{O}(\varepsilon
^{6}), \\
\beta _{+}^{\prime } &=&\mathcal{O}(\varepsilon ^{5}),\text{ }\gamma
_{+}^{\prime }=\mathcal{O}(\varepsilon ^{5}).
\end{eqnarray*}

\subsection{First component of $\mathcal{M}_{g}(\widehat{A_{0}},\widehat{%
C_{0}})$}

The first component is now the sum of small terms linear in $(\widehat{A_{0}}%
,\widehat{C_{0}})$ plus quadratic terms and terms independent of $(\widehat{%
A_{0}},\widehat{C_{0}})$ which tend exponentially to $0$ as $e^{\varepsilon
\delta _{\ast }x}$ for $x\rightarrow -\infty $ and $e^{-\sqrt{2}\varepsilon
x}$ for $x\rightarrow +\infty :$%
\begin{equation}
\mathcal{M}_{g}(\widehat{A_{0}},\widehat{C_{0}})|_{1}=-k_{-}\widehat{A_{0}}%
^{\prime \prime }+\frac{k_{-}^{2}}{4}\widehat{A_{0}}+\widehat{\phi _{0}}%
+\varphi _{1}(k_{-})  \label{1st component}
\end{equation}%
with%
\begin{eqnarray}
\varphi _{1}(k_{-}) &=&-k_{-}(A_{\ast }^{\prime \prime }+\alpha _{-}\chi
_{-}^{\prime \prime })+\frac{k_{-}^{2}}{4}(A_{\ast }-\chi _{-})+\alpha
_{-}\chi _{-}^{(4)}  \label{phi1} \\
&&-3(1-A_{\ast }^{2})\alpha _{-}\chi _{-}+gB_{\ast }^{2}\alpha _{-}\chi
_{-}+2gA_{\ast }B_{\ast }(\beta _{-}\chi _{-}+\beta _{+}\chi _{+}),  \notag
\\
\widehat{\phi _{0}} &=&\widetilde{\phi _{0}}(\omega x,\varepsilon ,k_{-},%
\widetilde{X},\widetilde{Y})-\chi _{-}\widetilde{\phi _{0}}(\omega
x,\varepsilon ,k_{-},\widetilde{X}^{(-\infty )},\widetilde{Y}^{(-\infty )}).
\notag
\end{eqnarray}%
More precisely we have, from (\ref{phi00tildea}), and taking into account (%
\ref{estim rest 1})%
\begin{eqnarray}
\widehat{\phi _{0}} &=&3[\alpha _{-}^{2}(A_{\ast }\chi _{-}^{2}-\chi
_{-})+2\alpha _{-}A_{\ast }\chi _{-}\widehat{A_{0}}+A_{\ast }\widehat{A_{0}}%
^{2}]+\alpha _{-}^{3}(\chi _{-}^{3}-\chi _{-})  \label{phi0hat} \\
&&+3\alpha _{-}^{2}\chi _{-}^{2}\widehat{A_{0}}+3\alpha _{-}\chi _{-}%
\widehat{A_{0}}^{2}+\widehat{A_{0}}^{3}+2gB_{\ast }[\alpha _{-}\chi _{-}%
\widehat{C_{0}}+(\beta _{-}\chi _{-}+\beta _{+}\chi _{+})\widehat{A_{0}}+%
\widehat{A_{0}}\widehat{C_{0}}]  \notag \\
&&+g(A_{\ast }+\alpha _{-}\chi _{-}+\widehat{A_{0}})[(\beta _{-}\chi
_{-}+\beta _{+}\chi _{+}+\widehat{C_{0}})^{2}+(\gamma _{+}\chi _{+}+\widehat{%
D_{0}})^{2}]  \notag \\
&&-\chi _{-}g(1+\alpha _{-})\beta _{-}^{2}+\widehat{f_{00}},  \notag \\
\widehat{f_{00}} &=&\sigma _{0}\varepsilon ^{2}k_{-}(A_{\ast }^{3}-\chi
_{-})+\mathcal{O}[\varepsilon ^{2+2/5}(e^{\varepsilon \delta _{\ast }x}\chi
_{(-\infty ,0)}+e^{-\varepsilon ^{1/5}\delta _{\ast }x}\chi _{(0,\infty
)})+\varepsilon ^{2}(|\widehat{X}|+|\widehat{Y}|)+\varepsilon |\widehat{D_{0}%
}^{\prime }|].  \notag
\end{eqnarray}%
We notice that for $\eta =\varepsilon \delta _{\ast }/2$ $(\eta <\varepsilon
\delta $ is necessary$),$ and due to Corollary \ref{Corollary stab},%
\begin{eqnarray*}
\frac{1}{\varepsilon ^{2}}\beta _{+}^{\prime } &=&\mathcal{O}(\varepsilon
^{3}),\text{ }\frac{1}{\varepsilon ^{2}}\gamma _{+}^{\prime }=\mathcal{O}%
(\varepsilon ^{3}), \\
||A_{\ast }^{\prime }||_{L_{\eta }^{2}} &=&\mathcal{O}(\varepsilon ^{1/10}),%
\text{ }||B_{\ast }^{\prime }||_{L_{\eta }^{2}}=\mathcal{O}(\varepsilon
^{1/2}), \\
||A_{\ast }^{\prime 2}||_{L_{\eta }^{2}} &=&\mathcal{O}(\varepsilon ^{7/10}),%
\text{ }||B_{\ast }^{\prime 2}||_{L_{\eta }^{2}}=\mathcal{O}(\varepsilon
^{3/2}), \\
||A_{\ast }^{\prime \prime }||_{L_{\eta }^{2}} &=&\mathcal{O}(\varepsilon
^{1/10}),\text{ }||B_{\ast }^{\prime \prime }||_{L_{\eta }^{2}}=\mathcal{O}%
(\varepsilon ^{3/2}).
\end{eqnarray*}%
Then, in using extensively $2|ab|\leq a^{2}+b^{2})$ and, for example%
\begin{equation*}
\frac{k_{-}^{2}}{4}||A_{\ast }-\chi _{-}||_{L_{\eta }^{2}}=\mathcal{O(}\frac{%
k_{-}^{2}}{\sqrt{\varepsilon }}),
\end{equation*}%
we obtain the estimates (here and in the following $c$ is a generic
constant, independent of $\varepsilon )$ 
\begin{eqnarray}
||\varphi _{1}(k_{-})||_{L_{\eta }^{2}} &\leq &c\left( \varepsilon
^{1/10}|k_{-}|+\frac{k_{-}^{2}+\varepsilon ^{4}}{\sqrt{\varepsilon }}+%
\widetilde{\omega }_{+}^{2}+\varepsilon ^{2}|\widetilde{\omega }%
_{+}|)\right) ,  \label{estimphi1a} \\
\int_{%
\mathbb{R}
}\varphi _{1}(k_{-})A_{\ast }^{\prime }dx &=&\mathcal{O}[(|k_{-}|+|%
\widetilde{\omega }_{+}|+\varepsilon ^{2})^{2}],  \notag
\end{eqnarray}%
using integration by parts and 
\begin{eqnarray*}
\int_{%
\mathbb{R}
}A_{\ast }^{\prime }A_{\ast }^{\prime \prime }dx &=&0, \\
\int_{%
\mathbb{R}
}(A_{\ast }-\chi _{-})A_{\ast }^{\prime }dx &=&\mathcal{O}(1) \\
\int_{%
\mathbb{R}
}(1-A_{\ast }^{2})A_{\ast }^{\prime }\chi _{-}dx &=&\mathcal{O}(1).
\end{eqnarray*}%
In next estimates, we use the following little Lemma (adapted from a simple
Sobolev inequality) where we notice that we loose one $\varepsilon ,$ due to
the weak exponential decay at $\infty :$

\begin{lemma}
\label{Sobolev}For any $u\in H_{\eta }^{1}$ and $\varepsilon $ sufficiently
small, we have%
\begin{equation*}
|u(x)|\leq c(||u||_{L_{\eta }^{2}}+\frac{1}{\varepsilon }||u^{\prime
}||_{L_{\eta }^{2}})
\end{equation*}%
where $c$ is independent of $\varepsilon .$
\end{lemma}

Then we may use%
\begin{eqnarray*}
|\widehat{A_{0}}^{(m)}(x)| &\leq &\frac{c}{\varepsilon }||\widehat{A_{0}}%
||_{H_{\eta }^{4}},\text{ }m=0,1,2,3 \\
|\widehat{C_{0}}^{(m)}(x)| &\leq &c\varepsilon ^{m}||\widehat{C_{0}}||_{%
\mathcal{D}_{1}},\text{ }m=0,1, \\
|\widehat{D_{0}}^{(m)}(x)| &\leq &c\varepsilon ^{m}||\widehat{D_{0}}||_{%
\mathcal{D}_{1}},\text{ }m=0,1.
\end{eqnarray*}%
Now, from $f_{00}$ in (\ref{phi0hat}), we have (see Remark \ref{remark estim
f00A'*}) 
\begin{equation*}
||d_{3}\varepsilon ^{2}A_{\ast }^{\prime \prime }+d_{4}\varepsilon
^{2}A_{\ast }^{2}A_{\ast }^{\prime \prime }+d_{5}\varepsilon ^{2}A_{\ast
}^{\prime \prime }B_{\ast }^{2}||_{L_{\eta }^{2}}=\mathcal{O}(\varepsilon
^{2+1/10}),
\end{equation*}%
and for example, from Lemma \ref{Sobolev}%
\begin{equation*}
2g||B_{\ast }\widehat{A_{0}}\widehat{C_{0}}||_{L_{\eta }^{2}}\leq c||%
\widehat{A_{0}}||_{H_{\eta }^{4}}||\widehat{C_{0}}||_{\mathcal{D}_{1}}\leq
c||(\widehat{A_{0}},\widehat{C_{0}})||_{\mathcal{D}_{0}}^{2},
\end{equation*}%
\begin{equation*}
||\widehat{A_{0}}^{2}||_{L_{\eta }^{2}}\leq \frac{c}{\varepsilon }||\widehat{%
A_{0}}||_{\mathcal{D}_{0}}^{2},\text{ \ }||\widehat{A_{0}}^{3}||_{L_{\eta
}^{2}}\leq \frac{c}{\varepsilon ^{2}}||\widehat{A_{0}}||_{\mathcal{D}%
_{0}}^{3}.
\end{equation*}%
We then obtain, for sufficiently small $\varepsilon ,|k_{-}|,|\widetilde{%
\omega }_{+}|,\widehat{A_{0}},\widehat{C_{0}},\widehat{D_{0}}$ in $%
\mathbb{R}
_{+}^{3}\times \mathcal{D}_{0}\times \mathcal{D}_{1}$ 
\begin{equation}
||\widehat{\phi _{0}}||_{L_{\eta }^{2}}\leq c\left( \varepsilon
^{2+1/10}+\varepsilon ^{3/2}|k_{-}|+\frac{k_{-}^{4}}{\sqrt{\varepsilon }}+%
\widetilde{\omega }_{+}^{4}+\frac{1}{\varepsilon }||\widehat{A_{0}}%
||_{H_{\eta }^{4}}^{2}+\frac{1}{\varepsilon ^{2}}||\widehat{A_{0}}%
||_{H_{\eta }^{4}}^{3}+||\widehat{C_{0}}||_{\mathcal{D}_{1}}^{2}+||\widehat{%
D_{0}}||_{\mathcal{D}_{1}}^{2}\right) .  \label{estim phi0a}
\end{equation}

\subsection{Second component of $\mathcal{M}_{g}(\widehat{A_{0}},\widehat{%
C_{0}})$}

For the second component we have 
\begin{equation}
\mathcal{M}_{g}(\widehat{A_{0}},\widehat{C_{0}})|_{2}=\frac{2\widetilde{%
\omega }_{+}}{\varepsilon }\widehat{D_{0}}^{\prime }+\widetilde{\omega }%
_{+}^{2}\widehat{C_{0}}+\widehat{\psi _{0r}}+\varphi _{2}(k_{-}),
\label{2nd component}
\end{equation}%
with%
\begin{eqnarray}
\varphi _{2}(k_{-}) &=&\widetilde{\omega }_{+}^{2}(B_{\ast }-\chi _{+})-%
\frac{1}{\varepsilon ^{2}}\beta _{-}\chi _{-}^{\prime \prime }-\frac{2}{%
\varepsilon ^{2}}\beta _{+}^{\prime }\chi _{+}^{\prime }-\frac{1}{%
\varepsilon ^{2}}\beta _{+}\chi _{+}^{\prime \prime }+\frac{2\widetilde{%
\omega }_{+}}{\varepsilon }\gamma _{+}\chi _{+}^{\prime }  \label{phi2} \\
&&-(3-gA_{\ast }^{2}-3B_{\ast }^{2})\beta _{+}\chi _{+}+[1-\chi
_{-}-g(A_{\ast }^{2}-\chi _{-})]\beta _{-}\chi _{-}+2gA_{\ast }B_{\ast
}\alpha _{-}\chi _{-},  \notag \\
\widehat{\psi _{0r}} &=&\widetilde{\psi _{0r}}(\omega x,\varepsilon ,k_{-},%
\widetilde{X},\widetilde{Y})-\chi _{+}\widetilde{\psi _{0r}}(\omega
x,\varepsilon ,k_{-},0,\widetilde{Y}^{(+\infty )})  \notag \\
&&-\chi _{-}\widetilde{\psi _{0r}}(\omega x,\varepsilon ,k_{-},\widetilde{X}%
^{(-\infty )},\widetilde{Y}^{(-\infty )}),  \notag
\end{eqnarray}%
where $\gamma _{+}=D_{0}^{(+\infty )}.$ For $\widehat{\psi _{0r}}$ we have 
\begin{eqnarray}
\widehat{\psi _{0r}} &=&2gA_{\ast }(\alpha _{-}\chi _{-}\widehat{C_{0}}%
+(\beta _{-}\chi _{-}+\beta _{+}\chi _{+})\widehat{A_{0}}+\widehat{A_{0}}%
\widehat{C_{0}})  \label{psi0hat} \\
&&+g(B_{\ast }+\beta _{+}\chi _{+}+\widehat{C_{0}})(\alpha _{-}^{2}\chi
_{-}^{2}+2\alpha _{-}\chi _{-}\widehat{A_{0}}+\widehat{A_{0}}^{2})  \notag \\
&&+g\beta _{-}\chi _{-}[(\alpha _{-}^{2}(\chi _{-}^{2}-1)+2\alpha _{-}\chi
_{-}\widehat{A_{0}}+\widehat{A_{0}}^{2}]  \notag \\
&&+[B_{\ast }(\beta _{-}\chi _{-}+\beta _{+}\chi _{+})^{2}-\chi _{+}\beta
_{+}^{2}]+[B_{\ast }(\gamma _{+}\chi _{+})^{2}-\chi _{+}\gamma _{+}^{2}] 
\notag \\
&&+\beta _{+}\chi _{+}(\chi _{+}^{2}-1)(\beta _{+}^{2}+\gamma
_{+}^{2})+\beta _{-}^{3}\chi _{-}(\chi _{-}^{2}-1)  \notag \\
&&+\widehat{C_{0}}[(\beta _{-}\chi _{-}+\beta _{+}\chi _{+}+\widehat{C_{0}}%
)^{2}+(\gamma _{+}\chi _{+}+\widehat{D_{0}})^{2}]  \notag \\
&&+2(B_{\ast }+\beta _{+}\chi _{+})(\beta _{-}\chi _{-}+\beta _{+}\chi _{+}%
\widehat{C_{0}}+\gamma _{+}\chi _{+}\widehat{D_{0}})  \notag \\
&&+(B_{\ast }+\beta _{-}\chi _{-}+\beta _{+}\chi _{+})(\widehat{C_{0}}^{2}+%
\widehat{D_{0}}^{2})+\widehat{g_{00r}},  \notag
\end{eqnarray}%
\begin{equation*}
\widehat{g_{00r}}=\mathcal{O}(\varepsilon ^{2+3/5}e^{\varepsilon \delta
_{\ast }x}\chi _{(-\infty ,0)}+\varepsilon ^{2+4/5}e^{-\varepsilon
^{1/5}\delta _{\ast }x}\chi _{(0,\infty )}+\varepsilon ^{2}(|\widehat{X}|+|%
\widehat{Y}|)+\varepsilon |\widehat{D_{0}}^{\prime }|).
\end{equation*}%
Now we use 
\begin{equation*}
||c_{5}\varepsilon ^{2}A_{\ast }A_{\ast }^{\prime \prime }B_{\ast
}||_{L_{\eta }^{2}}\leq c\varepsilon ^{2},
\end{equation*}%
and, as above%
\begin{equation*}
2g||A_{\ast }\widehat{A_{0}}\widehat{C_{0}}||_{L_{\eta }^{2}}\leq \frac{c}{%
\varepsilon }||(\widehat{A_{0}},\widehat{C_{0}})||_{\mathcal{D}_{0}}^{2},
\end{equation*}%
so that we obtain for sufficiently small $\varepsilon ,k_{-},\widetilde{%
\omega }_{+},\widehat{A_{0}},\widehat{C_{0}},\widehat{D_{0}}$ in $%
\mathbb{R}
^{3}\times \mathcal{D}_{0}\times \mathcal{D}_{1}$ (taking into account of (%
\ref{estim rest 1}))%
\begin{eqnarray}
||\widehat{\psi _{0r}}||_{L_{\eta }^{2}} &\leq &c\left( \varepsilon
^{2+1/10}+\frac{k_{-}^{4}+\widetilde{\omega }_{+}^{4}}{\sqrt{\varepsilon }}+%
\frac{1}{\varepsilon }||\widehat{A_{0}}||_{\mathcal{D}_{0}}^{2}+||\widehat{%
C_{0}}||_{\mathcal{D}_{1}}^{2}+||\widehat{D_{0}}||_{\mathcal{D}%
_{1}}^{2}\right)  \label{estim psi0r} \\
&&+c\left( (k_{-}^{2}+\widetilde{\omega }_{+}^{2})||(\widehat{A_{0}},%
\widehat{C_{0}})||_{\mathcal{D}_{0}}\right) ,  \notag
\end{eqnarray}%
In using, for example%
\begin{equation*}
||2gA_{\ast }B_{\ast }\alpha _{-}\chi _{-}||_{L_{\eta }^{2}}\leq c\frac{%
\widetilde{\omega }_{-}^{2}}{\sqrt{\varepsilon }},
\end{equation*}%
we obtain easily%
\begin{eqnarray}
||\varphi _{2}(k_{-})||_{L_{\eta }^{2}} &\leq &c(\frac{\widetilde{\omega }%
_{-}^{2}}{\sqrt{\varepsilon }}+\frac{(|\widetilde{\omega }_{+}|+\varepsilon
^{2})^{2}}{\varepsilon ^{2}}),  \label{estimphi2a} \\
\int_{%
\mathbb{R}
}\varphi _{2}(k_{-})B_{\ast }^{\prime }dx &=&\mathcal{O[}(k_{-}^{2}+%
\widetilde{\omega }_{+}^{2}+\varepsilon ^{4})],  \notag
\end{eqnarray}%
where the last estimates use%
\begin{eqnarray*}
\frac{1}{\varepsilon ^{2}}\int_{0}^{1}\beta _{+}^{\prime }\chi _{+}^{\prime
}B_{\ast }^{\prime }dx &=&\mathcal{O}(\varepsilon ^{4}) \\
\frac{1}{\varepsilon ^{2}}\int_{0}^{1}\beta _{+}\chi _{+}^{\prime \prime
}B_{\ast }^{\prime }dx &=&\mathcal{O}(|\widetilde{\omega }_{+}|+\varepsilon
^{2})^{2}
\end{eqnarray*}%
obtained, for the first integral in integrating by parts, and for the second
one in separating the oscillating part of order $\varepsilon ^{6}$ from the
constant part $\beta _{+}^{(c)}$ of $\beta _{+},$ for which we make an
integration by parts, in using $B_{\ast }^{\prime \prime }=\mathcal{O}%
(\varepsilon ^{2}B_{\ast }).$ More precisely we have%
\begin{eqnarray}
\int_{%
\mathbb{R}
}\varphi _{1}(k_{-})A_{\ast }^{\prime }dx+\int_{%
\mathbb{R}
}\varphi _{2}(k_{-})B_{\ast }^{\prime }dx &=&a_{2}\frac{k_{-}^{2}}{4}%
+a_{3}\sigma _{0}\varepsilon ^{2}k_{-}  \label{phi1+phi2} \\
&&+\mathcal{O}(|k_{-}^{3}|+\varepsilon ^{2}k_{-}^{2}+\widetilde{\omega }%
_{+}^{2}+\varepsilon ^{4}),  \notag
\end{eqnarray}%
with%
\begin{eqnarray*}
a_{2} &=&\int_{%
\mathbb{R}
}(A_{\ast }-\chi _{-})A_{\ast }^{\prime }dx-a_{3}, \\
2a_{3} &=&\int_{-1}^{0}\chi _{-}^{(4)}A_{\ast }^{\prime }-3\int_{%
\mathbb{R}
}(1-A_{\ast }^{2})A_{\ast }^{\prime }\chi _{-}dx+g\int_{%
\mathbb{R}
}(A_{\ast }B_{\ast }^{2})^{\prime }\chi _{-}dx,
\end{eqnarray*}%
We observe that (see Corollay \ref{corollary AA'})%
\begin{eqnarray*}
\int_{%
\mathbb{R}
}(A_{\ast }-\chi _{-})A_{\ast }^{\prime }dx &=&\frac{1}{2}+\mathcal{O}%
(\varepsilon ^{3/5}) \\
\int_{-1}^{0}\chi _{-}^{(4)}A_{\ast }^{\prime }dx &=&\mathcal{O}(\varepsilon
^{3/5}) \\
g\int_{-\infty }^{0}(A_{\ast }B_{\ast }^{2})^{\prime }\chi _{-}dx
&=&-g\int_{-1}^{0}(A_{\ast }B_{\ast }^{2})\chi _{-}^{\prime }dx=\mathcal{O}%
(\varepsilon ^{2/5})
\end{eqnarray*}%
\begin{equation*}
-3\int_{-\infty }^{0}(1-A_{\ast }^{2})\chi _{-}A_{\ast }^{\prime
}dx=3\int_{-1}^{0}(A_{\ast }-\frac{A_{\ast }^{3}}{3}-\frac{2}{3})\chi
_{-}^{\prime }dx=2+\mathcal{O}(\varepsilon ^{2/5}),
\end{equation*}%
so that%
\begin{eqnarray}
a_{2} &=&-3/2+\mathcal{O}(\varepsilon ^{2/5}),  \label{int a2} \\
a_{3} &=&4+\mathcal{O}(\varepsilon ^{2/5}).  \label{inta3}
\end{eqnarray}

\subsection{Component $\mathcal{L}_{g}\widehat{D_{0}}$}

For the third component we obtain%
\begin{equation}
\mathcal{L}_{g}\widehat{D_{0}}=-\frac{2\widetilde{\omega }_{+}}{\varepsilon }%
\widehat{C_{0}}^{\prime }+\widetilde{\omega }_{+}^{2}\widehat{D_{0}}+%
\widehat{\psi _{0i}}+\varphi _{3}(k_{-}),  \label{equ D0hat}
\end{equation}%
\begin{eqnarray*}
\varphi _{3}(\widetilde{\omega },k_{-},\omega x) &=&-\frac{2\widetilde{%
\omega }_{+}}{\varepsilon }[B_{\ast }^{\prime }+\beta _{-}\chi _{-}^{\prime
}+\beta _{+}\chi _{+}^{\prime }]-\frac{2}{\varepsilon ^{2}}\gamma
_{+}^{\prime }\chi _{+}^{\prime } \\
&&-\frac{1}{\varepsilon ^{2}}\gamma _{+}\chi _{+}^{\prime \prime
}-(1-gA_{\ast }^{2}-B_{\ast }^{2})\gamma _{+}\chi _{+},
\end{eqnarray*}%
and 
\begin{eqnarray*}
\widehat{\psi _{0i}} &=&\widetilde{\psi _{0i}}(\omega x,\varepsilon ,k_{-},%
\widetilde{X},\widetilde{Y})-\chi _{+}\widetilde{\psi _{0i}}(\omega
x,\varepsilon ,k_{-},0,\widetilde{Y}^{(+\infty )}) \\
&&-\chi _{-}\widetilde{\psi _{0i}}(\omega x,\varepsilon ,k_{-},\widetilde{X}%
^{(-\infty )},\widetilde{Y}^{(-\infty )}).
\end{eqnarray*}%
For sufficiently small $\varepsilon ,k_{-},\widetilde{\omega }_{+},\widehat{%
A_{0}},\widehat{C_{0}},\widehat{D_{0}}$ in $%
\mathbb{R}
^{3}\times \mathcal{D}_{0}\times \mathcal{D}_{1},$ we obtain the estimates%
\begin{equation}
||\varphi _{3}||_{L_{\eta }^{2}}\leq c(\varepsilon ^{3}+\frac{|\widetilde{%
\omega }_{+}|}{\sqrt{\varepsilon }}+\frac{|\widetilde{\omega }_{+}^{3}|}{%
\varepsilon }),  \label{estim phi3}
\end{equation}%
and taking into account of (\ref{estim rest 1}),%
\begin{eqnarray}
||\widehat{\psi _{0i}}||_{L_{\eta }^{2}} &\leq &c\{\varepsilon
^{1+1/10}+(k_{-}^{2}+\widetilde{\omega }_{+}^{2})||\widehat{D_{0}}||_{%
\mathcal{D}_{1}}+||\widehat{A_{0}}\widehat{D_{0}}||_{L_{\eta }^{2}}  \notag
\\
&&+||(\widehat{C_{0}}\widehat{D_{0}})||_{L_{\eta }^{2}}+||\widehat{D_{0}}||_{%
\mathcal{D}_{1}}^{2}\},  \label{estimpsi0i}
\end{eqnarray}%
where the term of order $\varepsilon ^{1+1/10}$ comes from%
\begin{equation*}
\varepsilon c_{9}||B_{\ast }A_{\ast }A_{\ast }^{\prime }||_{L_{\eta }^{2}}=%
\mathcal{O}(\varepsilon ^{1+1/10}).
\end{equation*}

\section{Bifurcation equation\label{sect: bifurc}}

Let us use an adapted Lyapunov-Schmidt method. Since%
\begin{equation*}
\mathcal{M}_{g}(A_{\ast }^{\prime },B_{\ast }^{\prime })=0,
\end{equation*}%
we now decompose $(\widehat{A_{0}},\widehat{C_{0}},\widehat{D_{0}})$ as%
\begin{eqnarray}
\widehat{A_{0}} &=&zA_{\ast }^{\prime }+u,  \label{new var} \\
\widehat{C_{0}} &=&zB_{\ast }^{\prime }+v,  \notag \\
\widehat{D_{0}} &=&w.  \notag
\end{eqnarray}%
For $\varepsilon $ small enough, the unknowns are now%
\begin{equation*}
(u,v)\in \mathcal{D}_{0},\text{ }w\in \mathcal{D}_{1},\text{ }(z,k_{-},%
\widetilde{\omega }_{+})\in 
\mathbb{R}
^{3}.
\end{equation*}

\begin{remark}
It might be interesting to give a physical interpretation of $z$. By
construction of the basic heteroclinic, it corresponds to a shift in $x$ of
the heteroclinic. However, $z$ occurs in the component $w$ which modifies
the phase of $B_{0}$ controlling the rolls parallel to the wall, themselves
affected by the slight change of wave length (due to $k_{+}).$ This "shift"
has no effect on the equilibrium at $-\infty $. We interpret this in saying
that the system of rolls parallel to the wall (in $x=0$), adapts itself to
fit with the rolls on the other side, orthogonal to the wall. Notice that $z$
corresponds to a "shift" of size of order $z/\varepsilon $ for the original
phase of the amplitude $B$ of rolls parallel to the wall.
\end{remark}

Then, equations (\ref{1st component},\ref{2nd component}) give ($Q_{0}$ is
the projection in $L^{2}$ on the range of $\mathcal{M}_{g}$) 
\begin{equation}
\mathcal{M}_{g}(u,v)=Q_{0}\binom{-k_{-}(zA_{\ast }^{\prime }+u)^{\prime
\prime }+\frac{k_{-}^{2}}{4}(zA_{\ast }^{\prime }+u)+\widehat{\phi _{0}}%
+\varphi _{1}(k_{-})}{\frac{2\widetilde{\omega }_{+}}{\varepsilon }w^{\prime
}+\widetilde{\omega }_{+}^{2}(zB_{\ast }^{\prime }+v)+\widehat{\psi _{0r}}%
+\varphi _{2}(k_{-})}.  \label{range1}
\end{equation}

\subsection{Resolution with respect to $\widetilde{\protect\omega }_{+}$ and 
$w$}

We observe that $(u,v)$ and $w$ appear non symmetrically, so we choose to
first solve equation (\ref{equ D0hat}), where the kernel of $\mathcal{L}_{g}$
is empty, and its range of codimension 1 (see Lemma \ref{Lem Linearoperator}%
). This has the advantage to give $w$ and $\widetilde{\omega }_{+}$ in
function of $(u,v,z,k_{-},\varepsilon ).$ So, let us start by solving the
compatibility condition.

Since%
\begin{equation*}
\int_{0}^{1}\frac{1}{\varepsilon ^{2}}\gamma _{+}^{\prime }\chi _{+}^{\prime
}B_{\ast }dx=-\int_{0}^{1}\frac{1}{\varepsilon ^{2}}\gamma _{+}(\chi
_{+}^{\prime }B_{\ast })^{\prime }dx=\mathcal{O}(\varepsilon ^{4}),
\end{equation*}%
and using Remark \ref{remark estim f00A'*}, we obtain the estimates 
\begin{eqnarray*}
\int_{%
\mathbb{R}
}\varphi _{3}B_{\ast }dx &=&-\frac{\widetilde{\omega }_{+}}{\varepsilon }[1+%
\mathcal{O}(|\widetilde{\omega }_{+}|+\varepsilon ^{2})^{2}]+\mathcal{%
O(\varepsilon }^{4})\text{ ,} \\
\int_{%
\mathbb{R}
}\widehat{\psi _{0i}}B_{\ast }dx &=&\mathcal{O}[\varepsilon
^{1+1/10}+(k_{-}^{2}+\widetilde{\omega }_{+}^{2})||\widehat{D_{0}}||_{%
\mathcal{D}_{1}}+||\widehat{D_{0}}||_{\mathcal{D}_{1}}^{2}+||\widehat{A_{0}}%
\widehat{D_{0}}||_{L_{\eta }^{2}}+||\widehat{C_{0}}\widehat{D_{0}}%
||_{L_{\eta }^{2}}] \\
&=&\mathcal{O}[\varepsilon ^{1+1/10}+\varepsilon ^{3/5}|z|||w||_{\mathcal{D}%
_{1}}+||(u,v)||_{\mathcal{D}_{0}}^{2}+||w||_{\mathcal{D}_{1}}^{2}+(k_{-}^{2}+%
\widetilde{\omega }_{+}^{2})||w||_{\mathcal{D}_{1}})].
\end{eqnarray*}%
Then the compatibility condition for equation (\ref{equ D0hat}) leads to%
\begin{equation*}
\frac{2\widetilde{\omega }_{+}}{\varepsilon }\int_{\mathbb{R}}B'_{*}B_{*}dx=\int_{\mathbb{R}
}\left[ -\frac{2\widetilde{\omega }_{+}}{\varepsilon }(zB_{\ast }^{\prime
\prime }+v^{\prime })+\widetilde{\omega }_{+}^{2}w+\widehat{\psi _{0i}}%
+\varphi _{3}\right] B_{\ast }dx,
\end{equation*}%
which gives%
\begin{eqnarray*}
\widetilde{\omega }_{+} &=&\int_{%
\mathbb{R}
}\left[ -2\widetilde{\omega }_{+}(zB_{\ast }^{\prime \prime }+v^{\prime
})+\varepsilon \widetilde{\omega }_{+}^{2}w\right] B_{\ast }dx \\
&&+\mathcal{O}[\varepsilon ^{2}+|\widetilde{\omega }_{+}|(|\widetilde{\omega 
}_{+}|+\varepsilon ^{2})^{2}+\varepsilon ^{1+2/5}|z|||w||_{\mathcal{D}_{1}}]
\\
&&+\varepsilon \mathcal{O}(||(u,v)||_{\mathcal{D}_{0}}^{2}+||w||_{\mathcal{D}%
_{1}}^{2}+(\widetilde{\omega }_{-}^{2}+\widetilde{\omega }_{+}^{2})||w||_{%
\mathcal{D}_{1}}).
\end{eqnarray*}%
The right hand side is a smooth function of its arguments, and may be solved
with respect to $\widetilde{\omega }_{+}$ (or equivalently with respect to $%
k_{+}$ since $\widetilde{\omega }_{+}\sim \frac{k_{+}}{2}$) by implicit
function theorem in the neighborhood of $0$ for%
\begin{equation*}
(u,v)\in \mathcal{D}_{0},\text{ }w\in \mathcal{D}_{1},\text{ }(\varepsilon ,%
\widetilde{\omega }_{-},z)\in 
\mathbb{R}
^{3},
\end{equation*}%
with%
\begin{equation*}
\widetilde{\omega }_{+}=\mathfrak{k}_{+}(\varepsilon ,\widetilde{\omega }%
_{-},z,(u,v),w)\in C^{1}(%
\mathbb{R}
^{3}\times \mathcal{D}_{0}\times \mathcal{D}_{1}).
\end{equation*}%
Moreover, we have the estimate 
\begin{equation}
|\mathfrak{k}_{+}|\leq c[\varepsilon ^{2}+\varepsilon ^{1+2/5}|z|||w||_{%
\mathcal{D}_{1}}+\varepsilon \widetilde{\omega }_{-}^{2}||w||_{\mathcal{D}%
_{1}}+\varepsilon (||(u,v)||_{\mathcal{D}_{0}}^{2}+||w||_{\mathcal{D}%
_{1}}^{2})].  \label{estim k+}
\end{equation}%
For solving equation (\ref{equ D0hat}) we now have%
\begin{equation*}
w=\mathcal{L}_{g}^{-1}[-\frac{2\mathfrak{k}_{+}}{\varepsilon }(zB_{\ast
}^{\prime \prime }+v^{\prime })+\mathfrak{k}_{+}^{2}w+\varphi _{3}+\widehat{%
\psi _{0i}}]
\end{equation*}%
which may be solved with respect to $w$ in $\mathcal{D}_{1},$ in the
neighborhood of $0$, by implicit function theorem, for 
\begin{equation*}
(\varepsilon ,k_{-},z,(u,v))\in 
\mathbb{R}
^{3}\times \mathcal{D}_{0}\text{ in a neighborhood of }0.
\end{equation*}%
Using (\ref{estim phi3}), (\ref{estimpsi0i}), (\ref{estim k+}) and%
\begin{equation*}
||\frac{B_{\ast }^{\prime \prime }}{\varepsilon }||_{L_{\eta }^{2}}=\mathcal{%
O}(\varepsilon ^{1/2}),\text{ \ }||\frac{v^{\prime }}{\varepsilon }%
||_{L_{\eta }^{2}}\leq ||v||_{\mathcal{D}_{1}},
\end{equation*}%
we obtain%
\begin{equation*}
w=\mathfrak{w}(\varepsilon ,\widetilde{\omega }_{-},z,u,v)
\end{equation*}%
with%
\begin{equation}
||\mathfrak{w}||_{\mathcal{D}_{1}}\leq c(\varepsilon ^{1+1/10}+\varepsilon
^{1/2}||(u,v)||_{\mathcal{D}_{0}}^{2}),  \label{estimw a}
\end{equation}%
and we deduce 
\begin{equation}
|\mathfrak{k}_{+}|\leq c(\varepsilon ^{2}+\varepsilon ||(u,v)||_{\mathcal{D}%
_{0}}^{2}).  \label{estim k+ a}
\end{equation}

\begin{remark}
\label{rmk w} The term of order $\varepsilon ^{1+1/10}$ in $\mathfrak{w}$ is 
$\varepsilon w_{1}+\mathcal{O}(\varepsilon ^{3/2})$ with $w_{1}$ coming from 
$\widehat{\psi _{0i}}$ and given by (see $\cite{Io23}$ for an explicit
formula of the pseudo-inverse of $\mathcal{L}_{g}$)%
\begin{equation}
w_{1}=c_{9}\mathcal{L}_{g}^{-1}[B_{\ast }A_{\ast }A_{\ast }^{\prime
}-2B_{\ast }^{\prime }\int_{%
\mathbb{R}
}B_{\ast }^{2}A_{\ast }A_{\ast }^{\prime }dx],\text{ }||w_{1}||_{\mathcal{D}%
_{1}}=\mathcal{O}(\varepsilon ^{1/10}),  \label{w1}
\end{equation}%
and the compatibility condition (orthogonality to $B_{\ast })$ is satisfied
with 
\begin{equation*}
||2B_{\ast }^{\prime }\int_{%
\mathbb{R}
}B_{\ast }^{2}A_{\ast }A_{\ast }^{\prime }dx||_{L_{\eta }^{2}}=\mathcal{O}%
(\varepsilon ^{1/10}).
\end{equation*}
\end{remark}

\subsection{Resolution with respect to $(u,v)$}

Now, we replace $w$ and $\widetilde{\omega }_{+}$ by their expressions $%
\mathfrak{w}$ and $\mathfrak{k}_{+},$ and consider (\ref{range1}) which may
be solved by implicit function theorem (by Lemma \ref{Lem Linearoperator}
the pseudo-inverse of $\mathcal{M}_{g}$ is bounded from $L_{\eta }^{2}$ to $%
\mathcal{D}_{0}$) with respect to $(u,v)$ in a neighborhood of $0$ in $%
\mathcal{D}_{0}$ for $(\varepsilon ,k_{-},z)$ close to $0$ in $\mathbb{R}%
^{3}.$ Indeed, the right hand side of (\ref{range1}) is smooth in its
arguments and assuming%
\begin{equation}
|k_{-}|<\varepsilon ,  \label{estimk-}
\end{equation}%
\begin{equation}
|z|<\varepsilon ^{1/5},  \label{estimz}
\end{equation}%
\begin{equation}
||u||_{\mathcal{D}_{0}}<\varepsilon ^{1+1/20},  \label{estim u,v a}
\end{equation}%
using (\ref{new var}) and collecting results of (\ref{1st component},\ref%
{estimphi1a},\ref{estim phi0a}) for the first component, and (\ref{2nd
component},\ref{estimphi2a},\ref{estim psi0r}) for the second component,
estimates in $L_{\eta }^{2}$ of the right hand side are as follows%
\begin{eqnarray*}
\text{1st comp.} &=&\mathcal{O}\left( \frac{k_{-}^{2}}{\sqrt{\varepsilon }}%
+\varepsilon ^{1/10}|k_{-}|+\varepsilon ^{2+1/10}+\varepsilon
^{7/10}z^{2}+|k_{-}||z|||u||_{\mathcal{D}_{0}}\right. \\
&&\left. +|z|\varepsilon ^{2/5}||(u,v)||_{\mathcal{D}_{0}}+\frac{1}{%
\varepsilon }||u||_{\mathcal{D}_{0}}^{2}+||v||_{\mathcal{D}%
_{1}}^{2}+(1/\varepsilon ^{2})||u||_{\mathcal{D}_{0}}^{3}\right) , \\
\text{2nd comp.} &=&\mathcal{O}\left( \varepsilon ^{2}+\frac{k_{-}^{2}}{%
\sqrt{\varepsilon }}+\varepsilon ^{3/2}|k_{-}|+\varepsilon ^{7/10}z^{2}+%
\frac{1}{\varepsilon }||u||_{\mathcal{D}_{0}}^{2}+||v||_{\mathcal{D}%
_{1}}^{2}\right. \\
&&\left. +(k_{-}^{2}+\varepsilon ^{2/5}|z|)||(u,v)||_{\mathcal{D}%
_{0}}\right) .
\end{eqnarray*}%
where we notice that, for example%
\begin{eqnarray*}
||\widehat{A_{0}}^{2}||_{L_{\eta }^{2}} &\leq &c(\varepsilon
^{7/10}z^{2}+|z|\varepsilon ^{2/5}||u||_{\mathcal{D}_{0}}+\frac{1}{%
\varepsilon }||u||_{\mathcal{D}_{0}}^{2}), \\
||\widehat{C_{0}}^{2}||_{L_{\eta }^{2}} &\leq &c(\varepsilon
z^{2}+|z|\varepsilon ||v||_{\mathcal{D}_{1}}+||v||_{\mathcal{D}_{1}}^{2}).
\end{eqnarray*}%
Applying implicit function theorem for $(\varepsilon ,k_{-},z)$ satisfying (%
\ref{estimk-},\ref{estimz}) in $%
\mathbb{R}
^{3},$ leads to%
\begin{equation*}
(u,v)=(\mathfrak{u,v)}(\varepsilon ,k_{-},z)\in \mathcal{D}_{0}
\end{equation*}%
with%
\begin{equation}
||(\mathfrak{u,v})||_{\mathcal{D}_{0}}\leq c(\varepsilon ^{2}+\frac{k_{-}^{2}%
}{\sqrt{\varepsilon }}+\varepsilon ^{1/10}|k_{-}|+\varepsilon ^{7/10}z^{2}),
\label{estim u,v}
\end{equation}%
which satisfies the a priori estimate (\ref{estim u,v a}). Now using (\ref%
{estimw a}), (\ref{estim k+ a}), (\ref{estimk-}), (\ref{estimz}) and (\ref%
{estim u,v}) we obtain 
\begin{eqnarray}
||\mathfrak{w}||_{\mathcal{D}_{1}} &\leq &c\varepsilon ^{1+1/10},
\label{estimw} \\
|\mathfrak{k}_{+}| &\leq &c\varepsilon ^{2},  \label{estimk+}
\end{eqnarray}%
where (\ref{estimk-}), (\ref{estimz}), (\ref{condition1}) and (\ref{basic
estim1}) need to be checked at the end. In fact we have the following

\begin{lemma}
\label{Lem conditions}Assuming that (\ref{estimk-}) and (\ref{estimz}) hold,
then (\ref{condition1}) and (\ref{basic estim1}) are satisfied.
\end{lemma}

\begin{proof}
Condition (\ref{condition1}) results immediately from the definition (\ref%
{new var}), Lemma \ref{Sobolev} and estimates (\ref{estim u,v}) and (\ref%
{estimw}). Then (\ref{basic estim1}) results from (\ref{changevariable}),
from the same estimates as above, and from (\ref{estimw}).
\end{proof}

\subsection{Final bifurcation equation \label{sect: finalbifurc}}

It remains to satisfy the orthogonality in $L^{2}$ of the right hand side of 
$\mathcal{M}_{g}(\widehat{A_{0}},\widehat{C_{0}})$ with $(A_{\ast }^{\prime
},B_{\ast }^{\prime })$ (see Lemma \ref{Lem Linearoperator}). This provides
one relationship, expressed as the cancelling of a function of $%
(z,k_{-},\varepsilon ),$ from which we extract the family of bifurcating
solutions. It gives%
\begin{eqnarray}
0 &=&\int_{%
\mathbb{R}
}[-k_{-}(zA_{\ast }^{\prime \prime \prime }+u^{\prime \prime })+\frac{%
k_{-}^{2}}{4}(zA_{\ast }^{\prime }+u)]A_{\ast }^{\prime }dx+\int_{%
\mathbb{R}
}(\widehat{\phi _{0}}+\varphi _{1})A_{\ast }^{\prime }dx  \notag \\
&&+\int_{%
\mathbb{R}
}[\frac{2\widetilde{\omega }_{+}}{\varepsilon }w^{\prime }+\widetilde{\omega 
}_{+}^{2}(zB_{\ast }^{\prime }+v)]B_{\ast }^{\prime }dx+\int_{%
\mathbb{R}
}(\widehat{\psi _{0r}}+\varphi _{2})B_{\ast }^{\prime }dx.  \label{compat}
\end{eqnarray}%
Let us define 
\begin{equation}
a_{1}=-\int_{%
\mathbb{R}
}A_{\ast }^{\prime \prime \prime }A_{\ast }^{\prime }dx=\int_{%
\mathbb{R}
}A_{\ast }^{\prime \prime 2}dx>0,\text{ }a_{1}=\mathcal{O}(\varepsilon
^{1/5})  \label{def a1}
\end{equation}%
so that, using Corollaries \ref{corollary AA'}, \ref{Corollary stab} and (%
\ref{estim u,v}), (\ref{estimw}), (\ref{estimk+}), we obtain%
\begin{eqnarray}
&&\int_{%
\mathbb{R}
}[-k_{-}(zA_{\ast }^{\prime \prime \prime }+u^{\prime \prime })+\frac{%
k_{-}^{2}}{4}(zA_{\ast }^{\prime }+u)]A_{\ast }^{\prime }dx  \notag \\
&=&a_{1}k_{-}z+\mathcal{O}(\varepsilon ^{2+1/10}|k_{-}|+\varepsilon
^{1/5}k_{-}^{2}+\frac{|k_{-}^{3}|}{\varepsilon ^{2/5}}+\varepsilon
^{4/5}|k_{-}|z^{2}),  \label{estimbifurc a}
\end{eqnarray}%
\begin{equation}
\int_{%
\mathbb{R}
}[\frac{2\widetilde{\omega }_{+}}{\varepsilon }w^{\prime }+\widetilde{\omega 
}_{+}^{2}(zB_{\ast }^{\prime }+v)]B_{\ast }^{\prime }dx=\mathcal{O}%
(\varepsilon ^{3+1/10}).  \label{estimbifurc b}
\end{equation}%
From (\ref{phi1+phi2}) we also have%
\begin{equation}
\int_{%
\mathbb{R}
}\varphi _{1}(k_{-})A_{\ast }^{\prime }dx+\int_{%
\mathbb{R}
}\varphi _{2}(k_{-})B_{\ast }^{\prime }dx=a_{2}\frac{k_{-}^{2}}{4}%
+a_{3}\sigma _{0}\varepsilon ^{2}k_{-}+\mathcal{O}(|k_{-}^{3}|+\varepsilon
^{2}k_{-}^{2}+\varepsilon ^{4}).  \label{estimbifurc c}
\end{equation}%
We have, from (\ref{phi0hat}), (\ref{psi0hat}), (\ref{estim u,v a}), (\ref%
{estimw}), (\ref{estimk+}), (\ref{estimk-}), (\ref{estimz}) and Remark \ref%
{remark estim f00A'*}%
\begin{equation}
\int_{%
\mathbb{R}
}\widehat{\phi _{0}}A_{\ast }^{\prime }dx=z^{2}[a_{0}^{\prime }+\mathcal{O}%
(\varepsilon ^{8/5})]+\sigma _{0}^{\prime }\varepsilon ^{2}k_{-}+\mathcal{O}%
[\varepsilon ^{2+4/5}+\varepsilon ^{3/2}k_{-}^{2}+\varepsilon
^{1/5}|z|(\varepsilon ^{2}+k_{-}^{2})],  \label{estim bifurc d}
\end{equation}%
with%
\begin{eqnarray}
a_{0}^{\prime } &=&\int_{%
\mathbb{R}
}(3A_{\ast }A_{\ast }^{\prime 3}+2gB_{\ast }B_{\ast }^{\prime }A_{\ast
}^{\prime 2}+gA_{\ast }A_{\ast }^{\prime }B_{\ast }^{\prime 2})dx+\mathcal{O}%
(\varepsilon ^{8/5})=\mathcal{O}(\varepsilon ^{4/5}),  \notag \\
\sigma _{0}^{\prime } &=&\sigma _{0}\int_{%
\mathbb{R}
}A_{\ast }^{\prime }(A_{\ast }^{3}-\chi _{-})dx+\mathcal{O}(\varepsilon
^{2+1/10})=\sigma _{0}[\frac{3}{4}+\mathcal{O}(\varepsilon ^{2/5})],
\label{def sigma'0}
\end{eqnarray}%
where (for example) the estimated term in $\varepsilon ^{2+4/5}$ comes from 
\begin{equation}
\varepsilon ^{2}(d_{2}-d_{4})\int_{%
\mathbb{R}
}A_{\ast }A_{\ast }^{\prime 3}dx\leq c\varepsilon ^{2+4/5},  \label{a'3}
\end{equation}%
occuring (see Remark \ref{remark estim f00A'*}) in $\int_{%
\mathbb{R}
}\widehat{f_{00}}A_{\ast }dx.$

We also obtain%
\begin{eqnarray}
\int_{%
\mathbb{R}
}\widehat{\psi _{0r}}B_{\ast }^{\prime }dx &=&z^{2}a_{0}^{\prime \prime }+%
\mathcal{O}(\varepsilon ^{3+1/5}+\varepsilon
^{2+3/5}|z|+k_{-}^{4}+\varepsilon ^{1+1/5}k_{-}^{2}  \label{estim bifurc e}
\\
&&+\varepsilon ^{11/20}|z|k_{-}^{2}+\varepsilon ^{1+3/20}|k_{-}||z|),  \notag
\end{eqnarray}%
with%
\begin{equation*}
a_{0}^{\prime \prime }=\int_{%
\mathbb{R}
}(gB_{\ast }B_{\ast }^{\prime }A_{\ast }^{\prime 2}+2gA_{\ast }A_{\ast
}^{\prime }B_{\ast }^{\prime 2}+B_{\ast }B_{\ast }^{\prime 3})dx+\mathcal{%
O(\varepsilon }^{1+3/4})=\mathcal{O}(\varepsilon ^{1+1/5}).
\end{equation*}%
Hence collecting (\ref{estimbifurc a}), (\ref{estimbifurc b}), (\ref%
{estimbifurc c}), (\ref{estim bifurc d}), (\ref{estim bifurc e}), and using
a priori estimates (\ref{estimk-}), (\ref{estimz}), we obtain the
bifurcation equation, in identifying main orders of independent coefficients,%
\begin{equation}
a_{0}z^{2}+a_{1}^{\prime }k_{-}z+a_{2}^{\prime }\frac{k_{-}^{2}}{4}%
+a_{3}^{\prime }\varepsilon ^{2}k_{-}+a_{4}\varepsilon
^{2+1/5}z+a_{5}\varepsilon ^{2+4/5}=0,  \label{bifurcationequation}
\end{equation}%
where we define 
\begin{equation}
a_{0}=a_{0}^{\prime }+a_{0}^{\prime \prime }+\mathcal{O}(\varepsilon
^{1+4/5})=\mathcal{O}(\varepsilon ^{4/5}).  \label{def a0}
\end{equation}%
Using Corollaries \ref{corollary AA'} and \ref{Corollary stab}, we notice
that the main contribution of this coefficient is precisely 
\begin{equation*}
a_{0}\sim \int_{-\infty }^{0}3A_{\ast }A_{\ast }^{\prime 3}dx=\mathcal{O}%
(\varepsilon ^{4/5}).
\end{equation*}%
From (\ref{def a1}), (\ref{phi1+phi2}), (\ref{def a0}) and (\ref{a'3}) we
obtain 
\begin{eqnarray}
a_{0} &=&\varepsilon ^{4/5}\overline{a_{0}}=\mathcal{O}(\varepsilon ^{4/5}) 
\notag \\
a_{1}^{\prime } &=&\int_{%
\mathbb{R}
}A_{\ast }^{\prime \prime 2}dx+\mathcal{O}(\varepsilon ^{1/2})=\mathcal{O}%
(\varepsilon ^{1/5})  \label{coef equ bifurc} \\
a_{2}^{\prime } &=&a_{2}+\mathcal{O}(\varepsilon ^{3/2})=-3/2+\mathcal{O}%
(\varepsilon ^{2/5})  \notag \\
a_{3}^{\prime } &=&a_{3}\sigma _{0}+\sigma _{0}^{\prime }+\mathcal{O}%
(\varepsilon ^{1/10})=\frac{19}{4}\sigma _{0}+\mathcal{O}(\varepsilon
^{1/10}),  \notag \\
a_{5}\varepsilon ^{4/5} &\sim &(d_{2}-d_{4})\int_{%
\mathbb{R}
}A_{\ast }A_{\ast }^{\prime 3}dx\sim \frac{(d_{2}-d_{4})}{3}a_{0}=\mathcal{O}%
(\varepsilon ^{4/5}).  \notag
\end{eqnarray}

The discriminant of the principal part of the quadratic form in $(z,k_{-})$
of the left hand side of (\ref{bifurcationequation}) is%
\begin{equation}
\Delta =a_{1}^{\prime 2}-a_{0}a_{2}^{\prime }=a_{1}^{\prime 2}+\mathcal{O}%
(\varepsilon ^{4/5})=\mathcal{O}(\varepsilon ^{2/5})  \label{discriminant}
\end{equation}%
which \emph{it is positive}. The bifurcation equation (\ref%
{bifurcationequation}) may then be written as%
\begin{equation}
\left( \frac{a_{2}^{\prime }k_{-}}{2}+a_{1}^{\prime }z+a_{3}^{\prime
}\varepsilon ^{2}\right) ^{2}-\Delta \left( z+\frac{a_{3}^{\prime \prime
}\varepsilon ^{2}}{\Delta }\right) ^{2}=-a_{2}^{\prime }a_{5}\varepsilon
^{2+4/5}+\mathcal{O}(\varepsilon ^{3+3/5})  \notag
\end{equation}%
where%
\begin{equation*}
a_{3}^{\prime \prime }=a'_1a_{3}^{\prime }-\frac{a_{4}a'_2}{2}\varepsilon ^{1/5}=\mathcal{O}(\varepsilon^{1/5}).
\end{equation*}%
Using the implicit function theorem, we obtain a family of solutions such
that $z$ and $k_{-}$ are given by (notice that $a'_1=\mathcal{O}(\varepsilon^{1/5})$)

i) if $a_{5}<0$%
\begin{eqnarray}
z &=&\sqrt{\frac{-3a_{5}}{2}}\frac{\varepsilon ^{1+2/5}}{a_{1}^{\prime }}%
\cosh \phi +\mathcal{O}(\varepsilon ^{1+2/5}).  \notag \\
k_{-} &=&2\sqrt{\frac{-2a_{5}}{3}}\varepsilon ^{1+2/5}\exp (-\phi )+\mathcal{%
O}(\varepsilon ^{1+3/5}).  \label{solution a5>0 Delta>0} \\
\phi &\in &%
\mathbb{R}
;  \notag
\end{eqnarray}

ii) if $a_{5}>0$%
\begin{eqnarray}
z &=&\frac{1}{a_{1}^{\prime }}\sqrt{\frac{3a_{5}}{2}}\varepsilon
^{1+2/5}\sinh \phi +\mathcal{O}(\varepsilon ^{1+2/5})  \notag \\
k_{-} &=&-2\sqrt{\frac{2a_{5}}{3}}\varepsilon ^{1+2/5}\exp (-\phi )+\mathcal{%
O}(\varepsilon ^{1+3/5})  \label{solution a5<0 Delta>0} \\
\phi  &\in &%
\mathbb{R}
.  \notag
\end{eqnarray}%
For $\varepsilon $ small enough, we notice that the principal part of the
solution only depends on $g$ and on coefficient $(d_{2}-d_{4})$ of the cubic
normal form (\ref{f0hat}). The above estimates on $u,v,w,z,k_{-}$ and Lemma %
\ref{Sobolev} imply that the conditions (\ref{estimk-}), (\ref{estimz}), are
satisfied for $\exp |\phi |\leq \varepsilon ^{-2/5}$. So, Lemma \ref{Lem
conditions} applies and Theorem \ref{theor walls} is then proved.

\begin{remark}
\label{link}It should be noted that the one parameter family of solutions
which are obtained for a fixed $\varepsilon $, correspond to convective
rolls at $-\infty $ with wave numbers 
\begin{equation*}
k_{c}(1+\varepsilon ^{2}k_{-})
\end{equation*}%
connected to convective rolls at $+\infty $ with wave numbers 
\begin{equation*}
k_{c}(1+2\varepsilon ^{2}\widetilde{\omega }_{+}).
\end{equation*}%
The calculations made above, show that we obtain $\widetilde{\omega }_{+}$
and $k_{-}$ as functions of $\varepsilon ,\phi $ where $\phi \in 
\mathbb{R}
$ such that $\exp |\phi |\leq \varepsilon ^{-2/5}.$ This is a one parameter
family of relationships between wave numbers at each infinity, depending on
the amplitude $\varepsilon ^{2}$ of rolls.
\end{remark}

\begin{remark}
\label{finalrmk} We might examine the limit size of $k_{-}.$ For example, is
it possible to obtain the case $k_{-}=k_{+}=2\widetilde{\omega }_{+}=%
\mathcal{O}(\varepsilon ^{2})?$ Then, looking at the bifurcation equation we
need to solve at main orders%
\begin{equation*}
(\overline{a_{0}}z^{2}+a_{5}\varepsilon ^{2})\varepsilon ^{4/5}=\mathcal{O}%
(\varepsilon ^{3+1/5}).
\end{equation*}%
Since $a_{5}\sim \frac{(d_{2}-d_{4})}{3}a_{0}$, this is only possible with $%
z\sim \varepsilon \sqrt{\frac{d_{4}-d_{2}}{3}}$ provided that 
\begin{equation*}
d_{4}-d_{2}>0,
\end{equation*}%
which coefficient of the cubic normal form (\ref{f0hat}) is a function of
the Prandtl number.
\end{remark}

\appendix

\section{Appendix}

\subsection{Reduction of the normal form\label{App1}}

We start with the N-S-B steady system of PDE's, applying spatial dynamics
with $x$ as "time" and considering solutions $2\pi /k$ periodic in $y$
(coordinate parallel to the wall). We show in \cite{BHI} that near
criticality a 12-dimensional center manifold reduction to a reversible
system applies for $(\mu ,k)$ close to $(0,k_{c}),$ where $\mu $ is $%
\mathcal{R}^{1/2}-\mathcal{R}_{c}^{1/2}$ ($\mathcal{R}$ is the Rayleigh
number), and $k_{c}$ the critical wave number. Then restricting the system
to solutions symmetric in $y$, the full system reduces to a 8-dimensional
one such as ($A_{0}$ (real) and $B_{0}$ are the amplitudes of the rolls
respectively at $x=-\infty ,$ and $x=+\infty ).$ Let us define

\begin{eqnarray*}
X &=&(A_{0},A_{1},A_{2},A_{3})^{t}\in 
\mathbb{R}
^{4}, \\
Y &=&(B_{0},B_{1})^{t}\in 
\mathbb{C}
^{2}, \\
k &=&k_{c}(1+\widetilde{k}),
\end{eqnarray*}%
so that the system may be written under normal form as (see \cite{BHI} )%
\begin{eqnarray}
\frac{dX}{dx} &=&LX+N(X,Y,\overline{Y},\mu ,\widetilde{k})+F(X,Y,\overline{Y}%
,\mu ,\widetilde{k}),  \label{full syst} \\
\frac{dY}{dx} &=&L_{k_{c}}Y+M(X,Y,\overline{Y},\mu )+G(X,Y,\overline{Y},\mu
),  \notag
\end{eqnarray}%
with%
\begin{eqnarray*}
LX &=&(A_{1},A_{2},A_{3},0)^{t}, \\
L_{k_{c}}Y &=&(ik_{c}B_{0}+B_{1},ik_{c}B_{1})^{t}.
\end{eqnarray*}%
The (reversible) system (\ref{full syst}) anticomutes with the symmetry $%
\mathbf{S}_{1}$ (representing the reflection $x\mapsto -x$)$.$ and commutes
with $\mathbf{\tau }_{\pi }$ (shift by half of one period in $y$ direction)$%
: $%
\begin{eqnarray*}
(A_{0},A_{1},A_{2},A_{3},B_{0},B_{1}) &\mapsto &^{\mathbf{S}%
_{1}}(A_{0},-A_{1},A_{2},-A_{3},\overline{B_{0}},-\overline{B_{1}}), \\
(A_{0},A_{1},A_{2},A_{3},B_{0},B_{1}) &\mapsto &^{\mathbf{\tau }_{\pi
}}(-A_{0},-A_{1},-A_{2},-A_{3},B_{0},B_{1}).
\end{eqnarray*}

\begin{remark}
We don't use the vertical symmetry $z\mapsto 1-z$ here (valid only in
rigid-rigid or free-free boundaries). In the case of rigid-free boundary
conditions, we have no such symmetry. The symmetry $\mathbf{\tau }_{\pi }$
implies that $F$ is odd in $X$ and $G$ even in $X$. Moreover it can be shown
that there is no term of degree 4 in $X,Y,\overline{Y}$ in the normal form.
\end{remark}

Then we obtain the estimates for $F$ and $G$ which are $C^{m}-$ smooth in
their arguments close to $0$, with $m$ as large as we need, and%
\begin{eqnarray}
|F(X,Y,\overline{Y},\mu ,\widetilde{k})| &\leq &c|X|(|X|^{2}+|Y|^{2}+|%
\widetilde{k}|+|\mu |)^{2}  \notag \\
|G(X,Y,\overline{Y},\mu )| &\leq &c(|X|^{2}+|Y|)(|X|^{2}+|Y|^{2}+|\mu |)^{2},
\label{estimG}
\end{eqnarray}%
and the normal form is (see\cite{BHI})%
\begin{equation*}
N(X,Y,\overline{Y},\mu )=\left( 
\begin{array}{c}
0 \\ 
A_{0}P_{1} \\ 
A_{1}P_{1}+c_{8}u_{8}+c_{13}u_{13} \\ 
A_{2}P_{1}+A_{0}P_{3}++c_{8}v_{8}+c_{13}v_{13}+d_{14}u_{14}%
\end{array}%
\right) ,
\end{equation*}%
\begin{equation*}
M(X,Y,\overline{Y},\mu )=\left( 
\begin{array}{c}
iB_{0}Q_{0}+\alpha _{10}u_{10} \\ 
iB_{1}Q_{0}+B_{0}Q_{1}+\alpha _{10}v_{10}+i\beta _{10}u_{10}+i\beta
_{12}u_{12}%
\end{array}%
\right) ,
\end{equation*}%
\begin{eqnarray}
P_{1} &=&b_{0}\mu +b_{0}^{\prime }\widetilde{k}%
+b_{1}u_{1}+b_{3}u_{3}+b_{5}u_{5}+b_{6}u_{6},  \notag \\
P_{3} &=&d_{0}\mu +d_{0}^{\prime \prime }\widetilde{k}%
^{2}+d_{1}u_{1}+d_{1}^{\prime }\widetilde{k}%
u_{1}+d_{3}u_{3}+d_{5}u_{5}+d_{6}u_{6},  \notag
\end{eqnarray}%
\begin{eqnarray*}
Q_{0} &=&\alpha _{0}\mu +\alpha _{1}u_{1}+\alpha _{3}u_{3}+\alpha
_{5}u_{5}+\alpha _{6}u_{6} \\
Q_{1} &=&\beta _{0}\mu +\beta _{1}u_{1}+\beta _{3}u_{3}+\beta
_{5}u_{5}+\beta _{6}u_{6},
\end{eqnarray*}%
where%
\begin{eqnarray}
u_{1} &=&A_{0}^{2},\text{ \ }v_{1}=A_{0}A_{1},\text{ \ }w_{1}=\frac{1}{2}%
A_{1}^{2},  \notag \\
u_{3} &=&2A_{0}A_{2}-A_{1}^{2},\text{ \ }v_{3}=3A_{0}A_{3}-A_{1}A_{2}  \notag
\\
u_{5} &=&B_{0}\overline{B_{0}},\text{ \ }v_{5}=\frac{1}{2}(B_{0}\overline{%
B_{1}}+\overline{B_{0}}B_{1}),\text{ \ }w_{5}=\frac{1}{2}B_{1}\overline{B_{1}%
}  \notag \\
u_{6} &=&i(B_{0}\overline{B_{1}}-\overline{B_{0}}B_{1}).  \notag
\end{eqnarray}%
\begin{eqnarray}
u_{8} &=&A_{0}v_{3}-A_{1}u_{3},\text{ \ }v_{8}=A_{1}v_{3}-2A_{2}u_{3}, 
\notag \\
u_{13} &=&A_{0}v_{5}-A_{1}u_{5},\text{ \ }v_{13}=A_{0}w_{5}-A_{2}u_{5}, 
\notag \\
u_{14} &=&A_{0}w_{5}+A_{2}u_{5}-A_{1}v_{5},  \notag
\end{eqnarray}%
\begin{eqnarray*}
u_{10} &=&B_{0}v_{1}-B_{1}u_{1},\text{ \ }v_{10}=2B_{0}w_{1}-B_{1}v_{1} \\
u_{12} &=&B_{0}v_{3}-B_{1}u_{3}.
\end{eqnarray*}

Then, the $X$ part of the system (\ref{full syst}) may be written as a 4th
order real ODE, while the $Y$ part becomes a 2nd order complex ODE as%
\begin{eqnarray*}
A_{0}^{(4)} &=&A_{0}[d_{0}\mu +(d_{0}^{\prime \prime }-b_{0}^{\prime 2})%
\widetilde{k}^{2}+d_{1}A_{0}^{2}+d_{1}^{\prime }\widetilde{k}A_{0}^{2}+d_{5}%
\widetilde{B_{0}}\overline{\widetilde{B_{0}}}+d_{1}^{\prime }\widetilde{k}%
A_{0}^{2} \\
&&+id_{6}(\widetilde{B_{0}}\overline{\widetilde{B_{0}}}^{\prime }-\overline{%
\widetilde{B_{0}}}\widetilde{B_{0}}^{\prime })]+(a_{0}\mu +3b_{0}^{\prime }%
\widetilde{k})A_{0}^{\prime \prime }+a_{1}A_{0}^{2}A_{0}^{\prime \prime
}+a_{2}A_{0}A_{0}^{\prime 2} \\
&&+a_{3}A_{0}\widetilde{B_{0}}^{\prime }\overline{\widetilde{B_{0}}}^{\prime
}+a_{4}A_{0}^{\prime }(\widetilde{B_{0}}\overline{\widetilde{B_{0}}}^{\prime
}+\overline{\widetilde{B_{0}}}\widetilde{B_{0}}^{\prime
})+a_{5}A_{0}^{\prime \prime }\widetilde{B_{0}}\overline{\widetilde{B_{0}}}
\\
&&+3ib_{6}A_{0}^{\prime \prime }(\widetilde{B_{0}}\overline{\widetilde{B_{0}}%
}^{\prime }-\overline{\widetilde{B_{0}}}\widetilde{B_{0}}^{\prime
})+a_{6}A_{0}A_{0}^{\prime }A_{0}^{\prime \prime \prime
}+a_{7}A_{0}A_{0}^{\prime \prime 2}+a_{8}A_{0}^{\prime 2}A_{0}^{\prime
\prime }+\mathcal{O}_{X}(5),
\end{eqnarray*}%
\begin{eqnarray*}
\widetilde{B_{0}}^{\prime \prime } &=&\widetilde{B_{0}}[\beta _{0}\mu +\beta
_{1}A_{0}^{2}+\beta _{5}\widetilde{B_{0}}\overline{\widetilde{B_{0}}}]+ic_{1}%
\widetilde{B_{0}}^{\prime }A_{0}^{2}+ic_{2}\widetilde{B_{0}}^{\prime }|%
\widetilde{B_{0}}|^{2}+ic_{3}\overline{\widetilde{B_{0}}}^{\prime }%
\widetilde{B_{0}}^{2} \\
&&+2i\alpha _{0}\mu \widetilde{B_{0}}^{\prime }+ic_{4}\widetilde{B_{0}}%
A_{0}A_{0}^{\prime }-2\alpha _{6}\widetilde{B_{0}}^{\prime }(\widetilde{B_{0}%
}\overline{\widetilde{B_{0}}}^{\prime }-\overline{\widetilde{B_{0}}}%
\widetilde{B_{0}}^{\prime }) \\
&&+c_{5}\widetilde{B_{0}}A_{0}A_{0}^{\prime \prime }+c_{6}\widetilde{B_{0}}%
A_{0}^{\prime 2}+c_{7}\widetilde{B_{0}}^{\prime }A_{0}A_{0}^{\prime }+ic_{8}%
\widetilde{B_{0}}A_{0}A_{0}^{\prime \prime \prime } \\
&&ic_{9}\widetilde{B_{0}}^{\prime }A_{0}A_{0}^{\prime \prime }+ic_{10}%
\widetilde{B_{0}}^{\prime }A_{0}^{\prime 2}+ic_{11}\widetilde{B_{0}}%
A_{0}^{\prime }A_{0}^{\prime \prime }+\mathcal{O}_{Y}(5),
\end{eqnarray*}%
with real coefficients $d_{j},d_{1}^{\prime }d_{0}^{\prime \prime
},a_{j},b_{j},b_{0}^{\prime },c_{j},\beta _{j},\alpha _{j}$ and%
\begin{equation}
\widetilde{B_{0}}=B_{0}e^{-ik_{c}x},\text{ }\widetilde{B_{1}}%
=B_{1}e^{-ik_{c}x},  \label{B0tilde}
\end{equation}%
\begin{eqnarray*}
d_{0} &=&-4k_{c}^{2}\beta _{0}>0,\text{ }d_{1}=-4k_{c}^{2}\beta _{5}<0,\text{
} \\
\frac{\beta _{1}}{\beta _{5}} &=&\frac{d_{5}}{d_{1}}:=g>0,\text{ }%
b_{0}^{\prime }=\frac{4k_{c}^{2}}{3},d_{0}^{\prime \prime }=-\frac{20}{9}%
k_{c}^{4},
\end{eqnarray*}%
\begin{eqnarray*}
\mathcal{O}_{X}(5) &=&\mathcal{O}(|X|(|X|^{2}+|Y|^{2}+\widetilde{k}^{2}+|\mu
|)^{2}), \\
\mathcal{O}_{Y}(5) &=&\mathcal{O[}(|X|^{2}+|Y|)(|X|^{2}+|Y|^{2}+|\mu |)^{2}],
\\
X &=&(A_{0},A_{0}^{\prime },A_{0}^{\prime \prime },A_{0}^{\prime \prime
\prime })^{t} \\
Y &=&(\widetilde{B_{0}},\widetilde{B_{0}}^{\prime }).
\end{eqnarray*}%
Notice that the high order rests $\mathcal{O}_{X}(5)$ and $\mathcal{O}%
_{Y}(5) $ are no longer autonomous, since they are functions of $e^{\pm
ik_{c}x}.$

Now, as indicated in \cite{BHI} we make the following scaling%
\begin{eqnarray}
x &=&\frac{1}{2\varepsilon k_{c}}\widetilde{x},\text{ }\mu =\frac{4k_{c}^{2}%
}{-\beta _{0}}\varepsilon ^{4},\text{ }\widetilde{k}=\varepsilon ^{2}k_{-}\ 
\label{scaling} \\
A_{0}(x) &=&\frac{2k_{c}}{\sqrt{\beta _{5}}}\varepsilon ^{2}\widetilde{A_{0}}%
(\widetilde{x}),\text{ }\widetilde{B_{0}}(x)=\frac{2k_{c}}{\sqrt{\beta _{5}}}%
\varepsilon ^{2}\widetilde{\widetilde{B_{0}}}(\widetilde{x}),  \notag
\end{eqnarray}%
so that the system above becomes, after suppressing the tildes,%
\begin{eqnarray}
A_{0}^{(4)} &=&k_{-}A_{0}^{\prime \prime }+A_{0}(1-\frac{k_{-}^{2}}{4}%
-A_{0}^{2}-g|B_{0}|^{2})+\widehat{f},  \notag \\
B_{0}^{\prime \prime } &=&\varepsilon ^{2}B_{0}(-1+gA_{0}^{2}+|B_{0}|^{2})+%
\widehat{g},  \label{new reduced syst}
\end{eqnarray}%
with additional cubic terms of the form (changing the definitions of
coefficients)%
\begin{eqnarray*}
\widehat{f} &=&id_{1}\varepsilon A_{0}(B_{0}\overline{B_{0}}^{\prime }-%
\overline{B_{0}}B_{0}^{\prime })+\sigma _{0}\varepsilon
^{2}k_{-}A_{0}^{3}+\varepsilon ^{2}[d_{3}A_{0}^{\prime \prime
}+d_{4}A_{0}^{2}A_{0}^{\prime \prime }+d_{2}A_{0}A_{0}^{\prime
2}+d_{6}A_{0}|B_{0}^{\prime }|^{2} \\
&&+d_{7}A_{0}^{\prime }(B_{0}\overline{B_{0}}^{\prime }+\overline{B_{0}}%
B_{0}^{\prime })+d_{5}A_{0}^{\prime \prime }|B_{0}|^{2}]+id_{8}\varepsilon
^{3}A_{0}^{\prime \prime }(B_{0}\overline{B_{0}}^{\prime }-\overline{B_{0}}%
B_{0}^{\prime })+\mathcal{O}(\varepsilon ^{4}),
\end{eqnarray*}%
\begin{eqnarray*}
\widehat{g} &=&\varepsilon ^{3}[ic_{0}B_{0}^{\prime }+ic_{1}B_{0}^{\prime
}|A_{0}|^{2}+ic_{2}B_{0}^{\prime }|B_{0}|^{2}+ic_{3}B_{0}^{2}\overline{B_{0}}%
^{\prime }+ic_{9}B_{0}A_{0}A_{0}^{\prime }] \\
&&+\varepsilon ^{4}[c_{4}B_{0}^{\prime }(B_{0}\overline{B_{0}}^{\prime }-%
\overline{B_{0}}B_{0}^{\prime })+c_{5}B_{0}A_{0}A_{0}^{\prime \prime
}+c_{6}B_{0}A_{0}^{\prime 2}+c_{7}B_{0}^{\prime }A_{0}A_{0}^{\prime }] \\
&&+\varepsilon ^{5}[ic_{8}B_{0}A_{0}A_{0}^{\prime \prime \prime
}+ic_{7}B_{0}^{\prime }A_{0}A_{0}^{\prime \prime }+ic_{10}B_{0}^{\prime
}A_{0}^{\prime 2}+ic_{11}B_{0}A_{0}^{\prime }A_{0}^{\prime \prime }+\mathcal{%
O}(\varepsilon ^{6}).
\end{eqnarray*}

\subsection{Equilibrium solution at $x=-\infty \label{App3}$}

Let us look for equilibria of (\ref{reduced syst a}), which should
correspond to the convective rolls at $x=-\infty $ parallel to $x$ - axis.
Cancelling all derivatives with respect to $x,$ we obtain a system commuting
with the symmetry $(A_{0},B_{0})\mapsto (A_{0},\overline{B_{0}}).$ It then
results a system of 2 real equations for $A_{0},B_{0}:$%
\begin{eqnarray*}
A_{0}(1-\frac{k_{-}^{2}}{4}-A_{0}^{2}+\sigma _{0}\varepsilon
^{2}k_{-}A_{0}^{2}-gB_{0}^{2})+\mathcal{O}(\varepsilon ^{4}) &=&0 \\
B_{0}(-1+gA_{0}^{2}+B_{0}^{2})+\mathcal{O}(\varepsilon ^{4}) &=&0,
\end{eqnarray*}%
where we may observe that the terms $\mathcal{O}(\varepsilon ^{4})$ in the
second equation contain at least terms of degree 1 in $B_{0},$ since they
come from terms of order 5 in $(A_{0},B_{0},\overline{B_{0}}).$ The first
terms not containing $B_{0}$ may be found at order 6 in $A_{0},$ which makes
order $\varepsilon ^{6}$ after the scaling (\ref{scaling}) in the rest
(12-6=6).

It then results that the equilibrium that we are looking for satisfies (by
implicit function theorem)%
\begin{eqnarray*}
A_{0}^{2} &=&1-\frac{k_{-}^{2}}{4}+\sigma _{0}\varepsilon ^{2}k_{-}+\mathcal{%
O}(\varepsilon ^{2}|k_{-}|^{3}+\varepsilon ^{4}), \\
B_{0} &=&\mathcal{O}(\varepsilon ^{6}).
\end{eqnarray*}

\begin{remark}
In the cases where vertical symmetry $z\mapsto 1-z$ applies, the additional
symmetry $S_{0}$ changes the signs of $A_{0}$ and $B_{0},$ implying that $%
Y=0 $ is an invariant subspace, so that in such cases $B_{0}=0$ for the
equilibrium at $-\infty .$
\end{remark}

\subsection{Periodic solution in $M_{+}$\label{App2}}

Let us consider the 4-dimensional reversible vector field corresponding to
the system (\ref{full syst}) with $X=0$ and rescaled. We intend to give
precise estimates on the family of periodic bifurcating solutions $%
B_{0}^{(+\infty )}(k_{+},x)$, here corresponding to the periodic convecting
rolls at infinity in $M_{+}$ with wave numbers close to $k_{c}$ (becomes $%
1/2\varepsilon $ after the scaling (\ref{scaling})).

Since we use the normal form up to cubic order, and since there is no term
of order 4, it takes the form (after the scaling used in \cite{BHI}, but
before we incorporate $e^{\frac{ix}{2\varepsilon }}$ in $B_{0},$ so that the
system is still autonomous):%
\begin{eqnarray}
\frac{dB_{0}}{dx} &=&\frac{i}{2\varepsilon }B_{0}+B_{1}+i\varepsilon
^{3}B_{0}P+\mathcal{\varepsilon }^{7}g_{0}(\varepsilon ,Y,\overline{Y})
\label{1-1 resonance syst} \\
\frac{dB_{1}}{dx} &=&\frac{i}{2\varepsilon }B_{1}+\varepsilon
^{2}B_{0}Q+i\varepsilon ^{3}B_{1}P+\mathcal{\varepsilon }^{6}g_{1}(%
\varepsilon ,Y,\overline{Y}),  \notag
\end{eqnarray}%
with%
\begin{eqnarray*}
Y &=&(B_{0},B_{1}) \\
P &=&\alpha +\beta |B_{0}|^{2}+\varepsilon \gamma K \\
Q &=&-1+|B_{0}|^{2}+\varepsilon \delta K \\
K &=&\frac{i}{2}(B_{0}\overline{B_{1}}-\overline{B_{0}}B_{1})
\end{eqnarray*}%
where we are looking for a periodic solution $(B_{0},B_{1})$, with wave
number $\omega $ close to $\frac{1+\varepsilon ^{2}k_{+}}{2\varepsilon }.$

\subsubsection{Principal part}

Let us first compute periodic solutions for $g_{0}=g_{1}\equiv 0.$ Then
these small terms will be perturbations treated by an adapted implicit
function theorem.

Without $g_{0}$ and $g_{1},$ let us use polar coordinates (see \cite{HIbook}
section 4.3.3) 
\begin{eqnarray*}
B_{0} &=&r_{0}e^{i\theta _{0}} \\
B_{1} &=&ir_{1}e^{i\theta _{1}}
\end{eqnarray*}%
then%
\begin{eqnarray*}
K &=&r_{0}r_{1}\cos (\theta _{0}-\theta _{1})=\text{const} \\
\frac{dr_{0}}{dx} &=&r_{1}\sin (\theta _{0}-\theta _{1}) \\
\frac{dr_{1}}{dx} &=&\varepsilon ^{2}r_{0}\sin (\theta _{0}-\theta
_{1})Q(\varepsilon ,r_{0}^{2},K) \\
r_{0}\frac{d\theta _{0}}{dx} &=&\frac{r_{0}}{2\varepsilon }+r_{1}\cos
(\theta _{0}-\theta _{1})+\varepsilon ^{3}r_{0}P \\
r_{1}\frac{d\theta _{1}}{dx} &=&\frac{r_{1}}{2\varepsilon }-\varepsilon
^{2}r_{0}\cos (\theta _{0}-\theta _{1})Q(\varepsilon
,r_{0}^{2},K)+\varepsilon ^{3}r_{1}P.
\end{eqnarray*}%
The required periodic solutions correspond to 
\begin{eqnarray*}
&&r_{0}\text{ and }r_{1}\text{ const} \\
\theta _{0} &=&\theta _{1},\text{ }\frac{d\theta _{0}}{dx}=\frac{%
1+\varepsilon ^{2}k_{+}}{2\varepsilon } \\
K &=&r_{0}r_{1},
\end{eqnarray*}%
hence%
\begin{eqnarray}
\frac{\varepsilon k_{+}}{2} &=&\frac{r_{1}}{r_{0}}+\varepsilon ^{3}P
\label{k+} \\
(\frac{r_{1}}{r_{0}})^{2} &=&-\varepsilon ^{2}Q.  \label{r1/r0}
\end{eqnarray}%
Solving (\ref{k+}) with respect to $r_{1}$ gives%
\begin{eqnarray*}
r_{1} &=&\varepsilon r_{0}\frac{k_{+}-2\varepsilon ^{2}(\alpha +\beta
r_{0}^{2})}{2(1+\varepsilon ^{4}\gamma r_{0}^{2})} \\
&=&\frac{\varepsilon r_{0}}{2}[k_{+}-2\varepsilon ^{2}(\alpha +\beta
r_{0}^{2})](1+\mathcal{O}(\varepsilon ^{4})),
\end{eqnarray*}%
and (\ref{r1/r0}) leads to%
\begin{equation*}
\frac{1}{4}[k_{+}-2\varepsilon ^{2}(\alpha +\beta r_{0}^{2})]^{2}+\frac{%
\varepsilon ^{2}\delta r_{0}^{2}}{2}[k_{+}-2\varepsilon ^{2}(\alpha +\beta
r_{0}^{2})]=(1-r_{0}^{2})(1+\gamma \mathcal{\varepsilon }^{4}r_{0}^{2})^{2}
\end{equation*}%
which is solved with respect to $r_{0}^{2},$ by implicit function theorem$:$%
\begin{eqnarray}
r_{0}^{2} &=&1-\frac{k_{+}^{2}}{4}+\mathcal{\sigma }_{1}\varepsilon
^{2}k_{+}+\sigma _{2}\varepsilon ^{4}+\mathcal{O}[(|k_{+}|+\varepsilon
^{2})^{4}],  \label{r0(k+,eps)} \\
r_{1} &=&\frac{\varepsilon r_{0}}{2}k_{+}+\mathcal{O}(\varepsilon ^{3}), 
\notag
\end{eqnarray}%
where we notice that coefficients $\sigma _{1}$ and $\sigma _{2}$ are
functions of the Prandtl number. We obtain a one-parameter family of
periodic solutions (parameter $k_{+}),$ with only the Fourier modes $e^{\pm
is}.$

\subsubsection{Estimates of higher order terms}

The proof below is new and self contained. There is a geometrical proof
without estimates in Iooss-P\'erou\`eme \cite{Io-Pe}, and a more precise
proof by Horn in \cite{Horn} section 3.5.

Let us define by $\omega $ the frequency of periodic solutions, where $%
\omega $ is close to 
\begin{equation*}
\omega _{0}=\frac{1+\varepsilon ^{2}k_{+}}{2\varepsilon },
\end{equation*}%
and set%
\begin{eqnarray*}
s &=&\omega x,\text{ }\omega =\omega _{0}+\widehat{\omega } \\
B_{0}(s) &=&r_{0}e^{is}+\widehat{B_{0}} \\
B_{1}(s) &=&ir_{1}e^{is}+i\widehat{B_{1}},
\end{eqnarray*}%
where $B_{0}$ and $B_{1}$ are $2\pi -$ periodic in $s,$ and $r_{0},r_{1}$
are solution of (\ref{k+},\ref{r1/r0}). Let us introduce the linear operator%
\begin{equation*}
L_{0}=\left( 
\begin{array}{cc}
-(i\omega _{0}\frac{d}{ds}+\frac{1}{2\varepsilon }+\varepsilon ^{3}P_{0}) & 
-1 \\ 
\varepsilon ^{2}Q_{0} & -(i\omega _{0}\frac{d}{ds}+\frac{1}{2\varepsilon }%
+\varepsilon ^{3}P_{0})%
\end{array}%
\right) ,
\end{equation*}%
acting in the function space $H^{1}(%
\mathbb{R}
/2\pi 
\mathbb{Z}
)\times L^{2}(%
\mathbb{R}
/2\pi 
\mathbb{Z}
).$ It appears that $L_{0}$ has a one-dimensional kernel 
\begin{equation*}
(r_{0}e^{is},r_{1}e^{is})\overset{def}{=}V_{0}e^{is}
\end{equation*}%
since (\ref{k+},\ref{r1/r0}) implies 
\begin{eqnarray*}
\lbrack (\omega _{0}-\frac{1}{2\varepsilon }-\varepsilon
^{3}P_{0}]r_{0}-r_{1} &=&0 \\
\varepsilon ^{2}Q_{0}r_{0}+[(\omega _{0}-\frac{1}{2\varepsilon }-\varepsilon
^{3}P_{0}]r_{1} &=&0,
\end{eqnarray*}%
with%
\begin{eqnarray*}
P_{0} &=&\alpha +\beta r_{0}^{2}+\varepsilon \gamma r_{0}r_{1}, \\
Q_{0} &=&-1+r_{0}^{2}+\varepsilon \delta r_{0}r_{1}.
\end{eqnarray*}%
Then the system (\ref{1-1 resonance syst}), to be completed by its complex
conjugate, becomes:%
\begin{eqnarray}
\widehat{\omega }V_{0}e^{is}+L_{0}\binom{\widehat{B_{0}}}{\widehat{B_{1}}}
&=&i\widehat{\omega }\frac{d}{ds}\binom{\widehat{B_{0}}}{\widehat{B_{1}}}%
+\left( 
\begin{array}{c}
\varepsilon ^{3}r_{0}P_{lin} \\ 
-\varepsilon ^{2}r_{0}Q_{lin}+\varepsilon ^{3}r_{1}P_{lin}%
\end{array}%
\right)  \notag \\
&&+\left( 
\begin{array}{c}
R_{0}(\widehat{Y},\overline{\widehat{Y}}) \\ 
R_{1}(\widehat{Y},\overline{\widehat{Y}})%
\end{array}%
\right) ,  \label{perturb equ}
\end{eqnarray}%
where%
\begin{eqnarray*}
P_{lin} &=&e^{2is}[\beta r_{0}\overline{\widehat{B_{0}}}+\frac{\varepsilon
\gamma }{2}(r_{0}\overline{\widehat{B_{1}}}+r_{1}\overline{\widehat{B_{0}}})]
\\
&&+[\beta r_{0}\widehat{B_{0}}+\frac{\varepsilon \gamma }{2}(r_{0}\widehat{%
B_{1}}+r_{1}\widehat{B_{0}})] \\
Q_{lin} &=&e^{2is}[-r_{0}\overline{\widehat{B_{0}}}+\frac{\varepsilon \delta 
}{2}(r_{0}\overline{\widehat{B_{1}}}+r_{1}\overline{\widehat{B_{0}}})] \\
&&+[-r_{0}\widehat{B_{0}}+\frac{\varepsilon \delta }{2}(r_{0}\widehat{B_{1}}%
+r_{1}\widehat{B_{0}})],
\end{eqnarray*}%
\begin{eqnarray*}
R_{0}(\widehat{Y},\overline{\widehat{Y}}) &=&\varepsilon
^{3}r_{0}e^{is}P_{quad}+\varepsilon ^{3}\widehat{B_{0}}%
(e^{-is}P_{lin}+P_{quad})-i\varepsilon ^{7}g_{0}, \\
R_{1}(\widehat{Y},\overline{\widehat{Y}}) &=&-\varepsilon
^{2}r_{0}e^{is}Q_{quad}-\varepsilon ^{2}\widehat{B_{0}}%
(e^{-is}Q_{lin}+Q_{quad}) \\
&&+\varepsilon ^{3}r_{1}e^{is}P_{quad}+\varepsilon ^{3}\widehat{B_{1}}%
(e^{-is}P_{lin}+P_{quad})-\varepsilon ^{6}g_{1},
\end{eqnarray*}%
with%
\begin{eqnarray*}
Q_{quad} &=&\widehat{B_{0}}\overline{\widehat{B_{0}}}+\frac{\varepsilon
\delta }{2}(\widehat{B_{0}}\overline{\widehat{B_{1}}}+\widehat{B_{1}}%
\overline{\widehat{B_{0}}}) \\
P_{quad} &=&\beta \widehat{B_{0}}\overline{\widehat{B_{0}}}+\frac{%
\varepsilon \gamma }{2}(\widehat{B_{0}}\overline{\widehat{B_{1}}}+\widehat{%
B_{1}}\overline{\widehat{B_{0}}}).
\end{eqnarray*}%
Let us decompose%
\begin{equation*}
\binom{\widehat{B_{0}}}{\widehat{B_{1}}}=\widehat{y}\binom{r_{1}e^{is}}{%
-r_{0}e^{is}}+\binom{\widetilde{B_{0}}}{\widetilde{B_{1}}}
\end{equation*}%
where $\widetilde{B_{0}}$ and $\widetilde{B_{1}}$ have no Fourier component
in $e^{is},$ and we take the component in $e^{is}$ orthogonal to $%
V_{0}e^{is} $, since adding a component proportional to $(r_{0},r_{1})$ is
equivalent to adapt $(r_{0},r_{1}).$

We first solve (\ref{perturb equ}) with respect to $(\widetilde{B_{0}},%
\widetilde{B_{1}})$ in using the implicit function theorem, since we observe
(notice the term $n\omega _{0}=\frac{n}{2\varepsilon }(1+\varepsilon
^{2}k_{+})$ in the operator for a Fourier component $e^{nis}$), that the
pseudo-inverse of $L_{0}$ is bounded from $H^{1}(\mathbb{R}/2\pi \mathbb{Z}%
)\times L^{2}(%
\mathbb{R}
/2\pi 
\mathbb{Z}
)$ to $H^{2}(%
\mathbb{R}
/2\pi 
\mathbb{Z}
)\times H^{1}(%
\mathbb{R}
/2\pi 
\mathbb{Z}
).$ Let us notice that the difference with the classical Hopf bifurcation
proof is that, norms in these spaces are chosen as, for example%
\begin{equation*}
||u||_{H^{2}}=\frac{1}{\varepsilon ^{2}}||u^{\prime \prime }||_{L^{2}}+\frac{%
1}{\varepsilon }||u^{\prime }||_{L^{2}}+||u||_{L^{2}},
\end{equation*}%
and notice that $H^{1}(\mathbb{R}/2\pi \mathbb{Z})$ is an algebra. It
results that we obtain an estimate such that%
\begin{equation*}
||(\widetilde{B_{0}},\widetilde{B_{1}})||_{H^{2}\times H^{1}}\leq
c(\varepsilon ^{2}|\widehat{y}|+\varepsilon ^{6}).
\end{equation*}%
It then remains to solve the 2-dimensional system in $(\widehat{\omega },%
\widehat{y})$ which is a real system, due to the reversibility symmetry:%
\begin{eqnarray*}
\widehat{\omega }r_{0}+\widehat{y}r_{1} &=&-\widehat{\omega }\widehat{y}%
r_{1}+\mathcal{O}(\varepsilon ^{4}|\widehat{y}|+\varepsilon ^{3}|\widehat{y}%
|+\varepsilon ^{7}) \\
\widehat{\omega }r_{1}-\widehat{y}r_{0} &=&\widehat{\omega }\widehat{y}r_{0}+%
\mathcal{O}(\varepsilon ^{3}|\widehat{y}|+\varepsilon ^{2}|\widehat{y}%
|+\varepsilon ^{6}),
\end{eqnarray*}%
which gives%
\begin{eqnarray*}
\widehat{\omega } &=&\mathcal{O}(\varepsilon ^{7}) \\
\widehat{y} &=&\mathcal{O}(\varepsilon ^{6}).
\end{eqnarray*}%
It results finally that the family of periodic solutions at $M_{+}$ are such
that%
\begin{eqnarray}
B_{0} &=&r_{0}e^{i\omega x}+\mathcal{O}(\varepsilon ^{6}),  \notag \\
B_{1} &=&ir_{1}e^{i\omega x}+\mathcal{O}(\varepsilon ^{6}),
\label{periodic sol +infty} \\
\omega &=&\frac{1}{2\varepsilon }+\frac{\varepsilon k_{+}}{2}+\mathcal{O}%
(\varepsilon ^{7}).  \notag
\end{eqnarray}

\end{document}